\documentclass[12pt]{article}

\usepackage{fullpage}
\usepackage{amsmath,amsthm,amssymb}
\usepackage{mathrsfs}
\usepackage[backref]{hyperref}
\hypersetup{colorlinks=true,
    linkcolor=red,citecolor=blue,
    urlcolor=blue,pdfstartview=FitH}
\usepackage{enumitem}
\usepackage{lineno}
\usepackage{cleveref}
\usepackage{mathtools}
\usepackage{color}
\usepackage{basicmathmacros}
\usepackage[normalem]{ulem} 
\usepackage{tikz}
\usepackage{float}
\usepackage{wrapfig}
\restylefloat{figure}
\usetikzlibrary{calc,intersections}
\usepackage{url}

\usepackage[inline,final]{showlabels} 
\newtheorem{theorem}{Theorem}[section]
\newtheorem{lemma}[theorem]{Lemma}
\newtheorem{corollary}[theorem]{Corollary}
\newtheorem{claim}[theorem]{Claim}

\theoremstyle{definition}
\newtheorem{definition}[theorem]{Definition}

\theoremstyle{remark}
\newtheorem{remark}[theorem]{Remark}

\usepackage[draft,inline,marginclue,index]{fixme}  
\fxsetup{theme=color,mode=multiuser}

\definecolor{color1}{RGB}{46,134,193}
\FXRegisterAuthor{iddo}{aiddo}{\colorbox{color1}{\color{white}Iddo}}
\FXRegisterAuthor{albert}{aalbert}{\colorbox{green}{\color{white}Albert}}

\newcommand{\commentout}[1]{}

\newcommand{\Formulas}{\mathrm{Fms}}
\newcommand{\Circuits}{\mathrm{Ckt}}
\newcommand{\coef}{\mathrm{coef}}
\newcommand{\polynomial}{\mathrm{polynomial}}
\newcommand{\evalboolean}{\ensuremath{\mathrm{eval}}}
\newcommand{\evalarithmetic}{\ensuremath{\mathrm{eval\text{-}arithmetic}}}
\newcommand{\degree}{\ensuremath{\mathrm{degree}}} 
\newcommand{\totaldegree}{\ensuremath{\mathrm{total\text{-}degree}}} 
\newcommand{\individualdegree}{\ensuremath{\mathrm{individual\text{-}degree}}} 
\newcommand{\maxindividualdegree}{\ensuremath{\mathrm{max\text{-}individual\text{-}degree}}}
\newcommand{\dimension}{\ensuremath{\mathrm{dimension}}}
\newcommand{\parametricdimension}{\ensuremath{\mathrm{parametric\text{-}dimension}}}

\newcommand{\decode}{\mathrm{dec}}
\newcommand{\encode}{\mathrm{enc}}

\newcommand{\HS}{{\ensuremath{\mathsf{HS}}}}
\newcommand{\BASIC}{{\ensuremath{\mathsf{BASIC}}}}

\newcommand{\Sonetwo}{{\ensuremath{\mathsf{S^1_2}}}}

\newcommand{\pind}{{\ensuremath{\mathsf{PIND}}}}
\newcommand{\lengthmax}{{\ensuremath{\mathsf{LENGTH{\text{-}}MAX}}}}
\newcommand{\LA}{\ensuremath{\mathcal{L}}}
\newcommand{\dWPHP}{\ensuremath{\mathsf{dWPHP}}}

\newcommand{\PV}{{\ensuremath{\mathsf{PV}}}}
\newcommand{\Log}{\mathrm{Log}}

\newcommand{\formalZ}{\mathbb{Z}}
\newcommand{\ComplexityFont}[1]{{\ensuremath{\mathsf{#1}}}}
\newcommand{\lang}[1]{{\ensuremath{\mathsf{#1}}}}
\newcommand{\func}[1]{{\ensuremath{\mathsf{#1}}}}
\newcommand{\newlang}[2]{\newcommand{#1}{\lang{#2}}}
\newcommand{\poly}{\func{poly}}

\newcommand{\wit}{\func{wit}}

\newcommand{\bin}{\mathrm{bin}}
\newcommand{\num}{\mathrm{num}}

\newcommand{\UniPoly}{\mathrm{UniPoly}}
\newcommand{\formvar}{\mathrm{x}}

\newcommand{\Img}{\mathrm{Img}}

\renewcommand{\P}{\ComplexityFont{P}}

\renewcommand{\E}{\ComplexityFont{E}}

\newcommand{\coNP}{\ComplexityFont{co\text{-}NP}}
\newcommand{\NP}{\ComplexityFont{NP}}
\newcommand{\FP}{\ComplexityFont{FP}}

\newcommand{\RP}{\ComplexityFont{RP}}
\newcommand{\BPP}{\ComplexityFont{BPP}}
\newcommand{\coRP}{\ComplexityFont{co\text{-}RP}}
\newcommand{\Ppoly}{\ComplexityFont{P\text{/}poly}}

\newcommand{\APEPP}{\ComplexityFont{APEPP}}
\newcommand{\SAPEPP}{\ComplexityFont{SAPEPP}}
\newcommand{\PtoNP}{\ComplexityFont{P}^{\ComplexityFont{NP}}}

\newcommand{\TFSigmatwoP}{\ComplexityFont{TF\Sigma_2P}}
\newcommand{\FSigmatwoP}{\ComplexityFont{F\Sigma_2P}}
\newcommand{\FPtoNP}{\ComplexityFont{FP}^{\ComplexityFont{NP}}}
\newcommand{\SIZE}{\ComplexityFont{SIZE}}

\newlang{\PIT}{PIT}
\newlang{\PITCONP}{Bounded\ PIT}
\newlang{\HITTINGSET}{Hitting\ Sets}
\newlang{\HITTINGSETCLASSES}{Hitting\ Sets\ for\ Circuit\ Classes}
\newlang{\HARD}{Hard\ Truth\ Table}
\newlang{\AVOID}{Avoid}
\newlang{\EMPTY}{Empty}

\newcommand{\av}{\va}
\newcommand{\bv}{\vb}
\newcommand{\cv}{\vc}
\newcommand{\dv}{\vd}
\newcommand{\pv}{\vp}

\newcommand{\zerov}{\ensuremath{\overline 0}}

\newcommand{\scriptC}{\mathscr{C}}

\newcommand{\SZ}{Schwartz-Zippel}

\title{\bf Feasibly Constructive Proof of Schwartz-Zippel Lemma \\
and the Complexity of Finding Hitting Sets} 

\author{Albert Atserias \\ Universitat Polit\`ecnica de Catalunya 
\and
Iddo Tzameret \\
Imperial College London}

\usepackage{graphicx}
\date{November 2024}

\begin{document}


\maketitle
\thispagestyle{empty} 

\begin{abstract}
The Schwartz-Zippel Lemma states that if a low-degree multivariate
polynomial with coefficients in a field is not zero everywhere in the
field, then it has few roots on every finite subcube of the
field. This fundamental fact about multivariate polynomials has found
many applications in algorithms, complexity theory, coding theory, and
combinatorics. We give a new proof of the lemma that offers some
advantages over the standard proof.

First, the new proof is more constructive than previously known
proofs. For every given side-length of the cube, the proof constructs
a polynomial-time computable and polynomial-time invertible surjection
onto the set of roots in the cube. The domain of the surjection is
tight, thus showing that the set of roots on the cube can be
compressed. Second, the new proof can be formalised in Buss' bounded 
arithmetic theory~$\Sonetwo$ for polynomial-time reasoning. One consequence
of this is that the theory~$\Sonetwo + \dWPHP(\PV)$ 
for  approximate counting can prove that the problem of verifying
polynomial identities (PIT) can be solved by polynomial-size circuits. 
The same theory can also
prove the existence of small hitting sets for any explicitly described
class of polynomials of polynomial degree. 

To complete the picture we show that the existence of such
hitting sets is \emph{equivalent} to the surjective weak pigeonhole 
principle~$\dWPHP(\PV)$, over the theory \Sonetwo. This is a contribution to a line 
of research studying the reverse mathematics of 
computational complexity (cf.~Chen-Li-Oliveira,~FOCS'24).  
One consequence of this is that
the problem of constructing small hitting sets for such classes is complete 
for the class APEPP of explicit construction problems whose totality follows 
from the probabilistic method (Kleinberg-Korten-Mitropolsky-Papadimitriou,~ITCS'21; 
cf.~\mbox{Korten},~FOCS'21). This class is also known and studied as the class 
of Range Avoidance Problems (Ren-Santhanam-Wang,~FOCS'22).
\end{abstract}

\newpage

\tableofcontents 
\thispagestyle{empty} 
\newpage
\clearpage
\pagenumbering{arabic}

\section{Introduction}

We study constructive proofs of the well-known Schwartz-Zippel Lemma, and their applications in complexity theory. This statement is sometimes referred to as the “Schwartz-Zippel-DeMillo-Lipton Lemma” following  
\cite{Zip79,Schwartz80,DeMilloLipton78}.
However, this result goes back to {\O}ystein Ore, 1922~\cite{ore} (over finite fields) and was subsequently rediscovered
by other authors; see \cite{arvind} who unearthed the history of this lemma. Due to this complicated history,  it is sometimes called ``The Polynomial Identity Lemma''.

\begin{theorem}[Schwartz-Zippel Lemma] 
\label{thm:SZ-meta}
Let~$\F$ be a field, let~$\vx$ be a set of~$n$ indeterminates, let~$S
\subseteq \F$ be a finite subset of field elements, and let~$P(\vx)$
be a multivariate polynomial in the indeterminates~$\vx$ with coefficients in the field~$\F$ 
and maximum individual degree at most~$d$.  Then, either every point in~$\F^n$ is a root of $P(\vx)$, 
or the number of roots in~$S^n$ is
at most~$d \cd n \cd |S|^{n-1}$. In particular, the number of roots in $S^n$ is either $|S|^n$ or
at most~$d \cd n \cd |S|^{n-1}$.\footnote{Our version of
the lemma is for \emph{maximum individual degree} and is
closely related to Zippel's version of the lemma in~\cite{Zip79}.
Zippel's bound, also for
maximum individual degree, is slightly tighter than 
the one stated here, namely:~$|S|^n \cdot (1 - (1 - d/|S|)^n)
\leq d \cd n \cd |S|^{n-1}$. There is also a better known
version of the lemma
for \emph{total degree} $D$, where the bound is $D \cdot |S|^{n-1}$.
Note that $d \leq D \leq d \cdot n$, so the bound for total degree 
is never further than a factor of $n$ from our bound.
In the regime where $|S|$ is bigger than $n$ and $d$, which is
the setting of most applications, all bounds are equally 
useful and yield similar conclusions. See~\cite{Regan} 
and the discussion there, for a comparison of the bounds, and 
for a discussion on how tight they are.}
\end{theorem}

The Schwartz-Zippel Lemma is a cornerstone in randomised algorithms and the use of randomness in computing, with a wide range of applications in computational complexity, coding theory, graph algorithms, and algebraic computation. The lemma shows that evaluating a non-zero polynomial on inputs chosen randomly from a large enough set is likely to find an input that produces a non-zero result. This offers a fast test with good guarantees for checking if a polynomial is identically zero.

While the lemma provides significant efficiency advantages in its applications, it is based on an existential statement---the existence of many non-roots for non-zero polynomials---without offering a deterministic way to find these non-roots or easily witness their absence. Accordingly, the standard textbook proof of the lemma, which goes by induction on the number of variables, although simple, does not reveal a feasible constructive argument. Specifically, it treats multivariate polynomials as potentially exponential sums of monomials. This non-constructive nature of the lemma is one reason why providing feasible constructive proofs has been challenging (cf.~\cite{LC11}).

Although different in nature, feasible constructivity in proofs and algorithms often go hand in hand, either informally or, at times, formally through translation theorems between proofs and computation. This work presents a new proof of the Schwartz-Zippel Lemma that fits within the framework of feasibly constructive proofs in a precise manner.  We then demonstrate  several applications of this proof, both in feasible constructive mathematics and in computational complexity, as we explain in the following. 
\bigskip

\noindent\textit{Organisation of the introduction.}
In \Cref{sec:intro:The Quest} we discuss  the  motivation behind seeking feasibly constructive proofs in general, and specifically the approach taken here to carry out the new proof within the formal logic framework of bounded arithmetic, which are formal theories accompanying and complementing the development of complexity theory. In \Cref{sec:intro:new proof of SZ} we describe (in the ``meta-theory'') the new feasibly constructive proof of Schwartz-Zippel. For those interested to first see  the new proof, it is presented at the end of \Cref{sec:intro:new proof of SZ}. We precede it with an exposition that aims to explain the intuition behind the new proof.
In \Cref{sec:intro:applications} we discuss applications of the new proof to bounded arithmetic, where our new argument helps to settle the open problem 
of  formalising  the Schwartz-Zippel Lemma in bounded arithmetic. 
Specifically, in \Cref{sec:intro:Formalisation of PIT/HS} we discuss how to prove the existence of small hitting sets and hence formalise Polynomial Identity Testing (PIT) in the theory. In \Cref{sec:intro:Contribution to Reverse math project} we describe a ``reversal'' theorem, in which the dual weak pigeonhole principle---namely, the statement that a function cannot map surjectively a small domain onto a large range---is shown to be equivalent to the existence of small hitting sets. And in \Cref{sec:intro:Applications to computational complexity and Range Avoidance Problem} we show that this  reversal theorem implies that finding hitting sets is complete for the class of Range Avoidance Problems (aka~$\APEPP$).

\subsection{The Quest for Feasibly Constructive Proofs}\label{sec:intro:The Quest}

While existence proofs, particularly those establishing the presence of certain combinatorial objects, are very useful---for instance, in probabilistic 
combinatorics, where one seeks to identify objects with certain properties that are as large or small as possible, such as expanders or combinatorial designs---it is widely acknowledged that constructive arguments, which provide explicit methods for constructing these objects, are often even more fruitful. This is especially significant in algorithmic contexts, where the utility of the objects (e.g., expanders) depends on their ``explicitness'', meaning that they must be feasibly computable, typically in polynomial time.

In computational complexity theory, the notions of constructivity and explicitness are equally critical. Fundamental questions about separating complexity classes hinge on explicit languages, such as the satisfiability problem (SAT), and whether SAT can be solved by polynomial-size circuits. For random, non-explicit languages, the analogous problems are almost trivially resolved: a simple counting argument shows that there exist non-explicit Boolean functions that cannot be computed by Boolean circuits of polynomial-size.

A related but distinct approach to feasible constructivity in complexity theory is found in the framework of \emph{bounded arithmetic}, which is the field that studies the computational complexity of concepts needed to prove different statements. In this setting, one aims to formalise constructivity in a more rigorous and systematic manner. Bounded arithmetic is a general name for a family of weak formal theories of
arithmetic (that is, natural numbers, and whose intended model is \N). These theories are characterized by their axioms and language (set of symbols), starting from a basic set of axioms providing the elementary properties of numbers. Each bounded
arithmetic theory possesses different additional axioms postulating the existence of different sets of numbers, or different kinds of induction  principles. Based on its specific axioms, each theory
of bounded arithmetic proves the totality of functions from different complexity classes (e.g.,
polynomial-time functions). We can typically consider such theories as working
over a logical language that contains the function symbols of that prescribed complexity class. In
this sense, proofs in the theory use concepts from a specific complexity class, and we can say that
the theory captures “reasoning in this class” (e.g., “polynomial-time reasoning”).


\smallskip 

In the current work we shall start with a standard naturally written constructive proof (in the ``meta-theory'') of the Schwartz-Zippel Lemma (\Cref{sec:intro:new proof of SZ}), following the first, less formal approach, to constructivity. We then show how the new proof fits in the formal approach to feasible constructivity of bounded arithmetic. We moreover exemplify the usual benefits of such proofs by showing applications both in bounded arithmetic and computational complexity.

\bigskip

\noindent\textbf{Background on theories of bounded arithmetic, their 
utility and applications.}
While the first theory for polynomial-time reasoning was the equational theory \PV\ considered
by Cook \cite{Coo75}, bounded arithmetic goes back to the work of Parikh \cite{Par71} and Paris-Wilkie \cite{PW85}. In a seminal work, Buss~\cite{Bus86} introduced other theories of bounded arithmetic and laid much of the foundation for future work in the field.

Using formal theories of bounded arithmetic is important for several reasons. First, it provides a rigorous framework to ask questions about provability, independence, and the limits of what can be proved by which means and arguments. The quest for ``barriers'' in computational complexity---namely, the idea that some forms of arguments are futile as a way to solve the fundamental hardness problems in complexity---such as Relativisation by Baker, Gill and Solovay \cite{BGS75}, Natural Proofs by Razborov and Rudich~\cite{RR97} or Algebrisation by Aaronson and Wigderson~\cite{AW09}, has played an important role in complexity theory. Nevertheless, the language of formal logic provides a more systematic framework for such a project (cf.~\cite{AIV92,IKK09} and the recent work \cite{CLO24}).
In that sense,  bounded arithmetic allows us to  identify suitable logical theories capable of formalising most contemporary complexity theory, and determine whether the major fundamental open problems in the field about lower bounds are provable or unprovable
within these theories.


Furthermore, bounded arithmetic serves as a framework in which the \emph{bounded reverse mathematics}
program is developed (in an analogy to Friedman and Simpson reverse mathematics program \cite{Sim99}). In this program, one seeks to find the weakest theory capable of proving a given theorem. In other words, we seek to discover what are the axioms that are not only sufficient to prove a certain theorem, but rather are also \emph{necessary}. 
Special theorems of interest are those of computer science and computational complexity
theory. The motivation is to shed light on the role of complexity classes in proofs, in the hope
to delineate, for example, those concepts that are needed to progress on major problems in computational
complexity from those that are not. For instance, it has been identified that apparently most
results in contemporary computational complexity can be proved using polynomial-time concepts
(e.g., in \PV) (cf.~\cite{Pic15}), and it is important to understand whether stronger theories and concepts are needed to prove certain results.

Accordingly, recent results in bounded arithmetic \cite{CLO24} seek to systematically build a bounded reverse mathematics of complexity lower bounds, in particular, as these are perhaps the most fundamental questions in complexity. This serves to establish complexity lower bounds as ``fundamental mathematical axioms'' which are equivalent, over the base theory, to rudimentary combinatorial principles such as the pigeonhole principle or related principles. Indeed, we will show that the (upper bound) statement about 
the existence of small hitting sets is equivalent to the (dual weak) pigeonhole principle. (It is interesting to note that the existence of explicit small hitting sets, which implies efficient PIT, is also a disguised lower bound statement as shown by Kabanets and Impagliazzo \cite{KabanetsImpagliazzo04}; though we do not attempt to formalise or pursue this direction in the current work.)


Another advantage of bounded arithmetic comes in the form of ``witnessing'' theorems. These are results that automatically convert formal proofs (of low logical-complexity theorems, namely, few quantifier alternations) in bounded arithmetic to feasible algorithms (usually, deterministic polynomial-time ones).
Witnessing theorems come in different flavours and strength, and recent work show the advantage of this method in both lower and upper bounds \cite{CarmosinoKabanetsKolokolovaOliveira2021,ILW23,LO23}. 

Moreover, the somewhat forbidding framework of bounded arithmetic forces one to think algorithmically from the get go, optimising  constructions. This resulted in new arguments to existing results, which proved very useful in complexity, beyond the scope of bounded arithmetic. A celebrated example is Razborov's new coding argument of the H{\aa}stad's switching lemmas \cite{Has89}, which emerged as work in bounded arithmetic \cite{Razb95}. (Intriguingly, our new argument for Schwartz-Zippel will also be based on a coding argument.)
\bigskip

\noindent\textbf{Randomness and feasibly constructive proofs.}\ 
One central part of complexity that was challenging to fit into bounded arithmetic is that of random algorithms. Randomness usually entails thinking of probability spaces of exponential size (e.g., the outcome of $n$ coin flips), and so cannot be directly used in most bounded arithmetic theories which cannot state the existence of exponential size sets. Initial work by Wilkie (unpublished) \cite[Theorem~7.3.7]{Kra95}, taken as well by Thapen \cite{Tha02}, and systematically developed in a  series of works of Je{\v{r}}{\'a}bek (cf.~\cite{Jer07}), concerned how to work with randomness in bounded arithmetic. 
Nevertheless, one of the most well-known examples for using  randomness in computing, the question of formulating Schwartz-Zippel and polynomial identity testing in particular, was left open due to the exponential nature of the standard argument (i.e., the need to treat polynomials as sums of exponential many monomials; see the first paragraph in Section~\ref{sec:intro:new proof of SZ}). For example, L\^{e} and Cook in \cite{LC11} list this as an open question (where they settled for formalising a special case of the \SZ\ Lemma).

The formalisation of \SZ\ Lemma and PIT in bounded arithmetic and its consequences we present, hopefully exemplify 
several of the benefits of bounded arithmetic described above. It serves to fill in the missing link in the 
formalisation of complexity of randomness; it produces a new (coding) argument of Schwartz-Zippel lemma that may be 
of independent interest; it establishes the existence of hitting sets as equivalent to the dual weak pigeonhole  
principle, and thus provides further examples of the ``axiomatic'' nature of building blocks in complexity theory; 
lastly, it has consequences to computational complexity, by showing that hitting sets are complete for the class of 
range avoidance problems.

\subsection{New Proof of the Schwartz-Zippel Lemma}
\label{sec:intro:new proof of SZ}

The standard \SZ\ proof proceeds by induction on the number of variables~$n$
but is not a feasibly constructive proof.
It is non-constructive in the sense that each inductive step involves 
stating properties of objects that are of exponential size. For
example, in step~$i$ of the induction, the proof states that the current 
polynomial on~$n-i+1$ variables can
be rewritten into a univariate polynomial in the last variable, 
with coefficients taken from the ring of polynomials in the 
first~$n-i$ variables. The non-constructive character of the proof 
stems then from the fact that there is no (known) efficient way of 
iterating this rewriting~$n$ times, unless the polynomial is given in 
explicit sum of monomials form. Note, however, that
the sum of monomials form is typically of exponential size.

Moshkovitz~\cite{Mos10} provided an alternative proof of the \SZ\ Lemma, but
only for finite fields. The proof in~\cite{Mos10} does not use explicit induction, 
but it is anyway unclear how to make it strictly constructive in our sense, 
namely how to identify polynomial-time algorithms for the concepts used in the 
proof, and then use these to formalise the whole argument in a relatively weak 
theory. We refer to Moshkovitz' proof again later in this section to compare it
with our approach.

\bigskip
\noindent\textit{Towards our new proof.}\ 
%
For the sake of exposition, let us begin with two natural but
failed attempts to a new proof. In what follows, let $P(\vx)$ be a 
polynomial over 
the field~$\F$, with~$n$ variables and maximum individual degree~$d$, 
and let~$\va \in \F^n$ be a given non-root;~$P(\va) \ne 0$. 
Let~$S \subseteq \F$ be a finite subset of the field.
%
%
In the \emph{first attempt}, we try to cover the cube~$S^n$ with at
most~$n\cdot |S|^{n-1}$ lines, each emanating from the non-root~$\va$.  
These lines are the subsets of~$\F^n$ 
of the form~$\{ \va + t \cd \vb : t \in
\F \}$ for some~$\vb \in \F^n$. For each such line~$L$, 
note that~$P$ retricted to~$L$, defined as~$P_L(t) := P(\va + t\cd\vb)$, is a
non-zero \emph{univariate} polynomial of degree at most~$d$. It is
non-zero because it evaluates to~$P(\va) \not= 0$ at~$t =
0$, and it has degree at most~$d$ because it is a linear restriction of~$P$.
Therefore, by the fundamental theorem for 
(univariate) polynomials, each such line has at most~$d$ roots, and since
the lines cover~$S^n$, we count 
at most~$d \cdot n \cdot |S|^{n-1}$ roots in~$S^n$ in total, and we are done.

The problem with this approach is that it is not always 
possible to cover $|S|^n$ with at most~$n \cdot |S|^{n-1}$ lines emanating 
from a single point~$\va$. A simple counterexample can be found already
at dimension $n = 2$ with $S = \{0,1,\ldots,q-1\}$ and $\va = (0,0)$: the $2q-1$ points
of the form $(i,q-1)$ or $(q-1,i)$ with $i \in S$ require each its own line
emanating from the origin, and the remaining $q^2-2q-2$ points cannot
be covered with one more line.

In the \emph{second attempt}, we want to cover the cube~$S^n$ again with at
most~$n \cdot |S|^{n-1}$ lines, but now we try with \emph{parallel lines}. 
For example we could consider
the set of axis-parallel lines of the form $\{(c_1,\ldots,c_{k-
1},t,c_{k+1},\ldots,c_n) : t \in \F \}$ with~$c_j \in S$ for all $j \not= k$, 
for some fixed~$k = 1,\ldots,n$.
The problem with this approach now is that it is not clear that all
such lines will go through some non-root of~$P$.
It is tempting to consider some sets of parallel lines that
are more related to the given non-root~$\va$, such as 
the set of lines of the form~$\{ \vb + t \cdot \va : t \in \F\}$ 
whose \emph{gradient} is $\va$. This is indeed the approach
taken by Moshkovitz in~\cite{Mos10}, but as far as we can see this
does not work for arbitrary non-roots~$\va$, and works only for a specific
kind of non-roots that seem hard to find in the first place.
\bigskip

\noindent\textit{Our approach.}
We are now ready to explain the two new ideas that we use in our proof, 
and how they overcome the obstacles of the previous two attempts.
First, instead of \emph{covering} the roots in $S^n$ 
with~$n \cd |S|^{n-1}$ lines where the polynomial is non-zero, we 
are going to \emph{encode} each root in~$S^n$
with one in~$n \cd |S|^{n-1}$ such lines, together with an additional
number~$i$ in~$1,\ldots,d$.
The lines we use to encode the roots are not
necessarily pair-wise parallel, though each line
will be parallel to one of the axes.
%
Second, to actually find this line, our proof uses a 
\emph{hybrid-type argument}. Concretely, to 
encode the root~$\cv$, we start at 
a line through~$\av$ and end at a line through~$\cv$.
Along the way, the process travels across at
most~$n$ axis-parallel lines of the cube~$S^n$,
changing the dimension of travel at each step.
The hybrid-argument is used to
preserve the property that the restriction
of~$P$ to the current line is still a non-zero polynomial. 
The exact details of how this is done are explained below.
\bigskip


\fboxsep=6pt
\begin{center}
\fbox{\begin{minipage}[c]{0.92\textwidth} 
\setlength{\parindent}{1.5em}
\noindent\textbf{New proof
      of \SZ\ lemma}: 
%
%
 Let~$\va = (a_1,\ldots,a_n) \in \mathbb{F}^n$ be such that~$P(\va) \not= 0$,
 and let~$S$ be a finite subset of the field~$\mathbb{F}$.
 Each vector~$\vc = (c_1,\ldots,c_n) \in S^n$ of field elements in~$S$ can 
 be encoded with~$n$ numbers, each from~$1$ to~$|S|$, 
 by identifying each~$c_i$ with its position in an arbitrary ordering
 of the finite set~$S$.
 Our goal is to \emph{encode the roots of~$P$ in~$S^n$ with shorter codewords}. 
 To achieve this we will use only the fact that we know a non-root~$\va \in \mathbb{F}^n$. 
 
 Let~$\vc = (c_1,\ldots,c_n) \in S^n$ be such that~$P(\vc) = 0$.
    Find the minimal index~$k$ between~$1$ and~$n$ such
    that~$P(c_1,\dots,c_{k-1},a_k,a_{k+1},\dots,a_n)\neq 0$
    while~$P(c_1,\dots,c_{k-1},c_{k},a_{k+1},\dots, a_n)= 0$.  
    Such a~$k$ must exist since by assumption~$P(\va) \not= 0$ and~$P(\cv) = 0$.
    Observe that, if we are given both~$\va$ and $\cv$, then finding~$k$ is easy by looping through 
    the coordinates from left to
    right.  The argument hinges on
    the following observation: 
    \begin{quote}
    \emph{Knowing this~$k$ allows us to
    encode, given~$c_1,\dots,c_{k-1},c_{k+1},\dots,c_{n}$, the
    field element~$c_k\in S$ by a single number~$i$
    between~$1$ and~$d$.} 
    \end{quote}
    Therefore, we can use~$i$, together with~$k$ and the positions 
    of~$c_1,\ldots,c_{k-1},c_{k+1},\ldots,c_n$ in the fixed ordering of~$S$,
    as a code for~$\cv$. This shows that
    the set of roots in~$S^n$ can be encoded using
    only~$d \cd n \cd |S|^{n-1}$ numbers
    instead of the trivial~$|S|^n$ bound, concluding the argument.  
    
    

    In detail, consider the
    root~$\vc=(c_1,\dots, c_{k-1},c_{k},c_{k+1}\dots,c_{n})$, where~$k$
    is the minimal index as above.  We encode~\vc\ by the~$n-1$
    elements~$c_1,\dots, c_{k-1},c_{k+1},\dots,c_{n}$ from~$S$, together
    with~$k$ and a second index~$i$ such that~$c_k \in S$ is
    the~$i$-th root in~$S$ of the univariate
    polynomial~$Q(t) = P(c_1,\dots,c_{k-1},t,a_{k+1},\dots,a_n)$. 
    This encoding works because given the numbers~$k$ and $i$, 
    and the elements~$(c_1,\dots, c_{k-1},c_{k+1}\dots,c_{n})$ we can recover~$\vc$. 
    Indeed, by the choice of~$k$ and~$i$, the univariate
    polynomial~$Q(t)$ is \emph{non-zero and has degree at most~$d$}, which
    means that it has no more than~$d$ roots in $\mathbb{F}$.
    By looping through~$S$ we can find
    its~$i$-th root in~$S$ which, by construction, is~$c_{k}$. The fact that each 
    root of~$P$ in~$S^n$ is coded by~$n-1$ of its components
    and the two positive integers~$k$ and~$i$ finishes
    the proof as it means that the total number of such roots is bounded above
    by the number of all possible such codes:~$d\cd n \cd |S|^{n-1}$. 
\end{minipage}}
\end{center}

\begingroup
\setlength{\columnsep}{10pt}
\setlength{\intextsep}{3pt}
\begin{wrapfigure}{r}{5.5cm}
%
%
\begin{tikzpicture}
   \clip (-3,-3.4) rectangle (3,1.2);
   \coordinate (tf) at (0,0);
   \coordinate (bf) at (0,-3);
   \coordinate (tr) at (15:2.5cm);
   \coordinate (tl) at (165:2.5cm);

   \coordinate (fr) at ($ (tf)!5!(tr) $);
   \coordinate (fl) at ($ (tf)!5!(tl) $);
   \coordinate (fb) at ($ (tf)!15!(bf) $);

   \path[name path=brpath] (bf) -- (fr);
   \path[name path=rbpath] (tr) -- (fb);
   \path[name path=blpath] (bf) -- (fl);
   \path[name path=lbpath] (tl) -- (fb);
   \path[name path=trpath] (tl) -- (fr);
   \path[name path=tlpath] (tr) -- (fl);

   \draw[name intersections={of=brpath and rbpath}] (intersection-1)coordinate (br){}; 
   \draw[name intersections={of=blpath and lbpath}] (intersection-1)coordinate (bl){}; 
   \draw[name intersections={of=trpath and tlpath}] (intersection-1)coordinate (tb){}; 

   \shade[right color=gray!10, left color=black!50, shading angle=105] (tf) -- (bf) -- (bl) -- (tl) -- cycle;
   \shade[left color=gray!10, right color=black!50, shading angle=75] (tf) -- (bf) -- (br) -- (tr) -- cycle;

  \begin{scope}
    \clip (tf) -- (tr) -- (tb) -- (tl) -- cycle;
    \shade[inner color = gray!5, outer color=black!50, shading=radial] (tf) ellipse (3cm and 1.5cm);
  \end{scope}

  \draw (tf) -- (bf);
  \draw (tf) -- (tr);
  \draw (tf) -- (tl);
  \draw (tr) -- (br);
  \draw (bf) -- (br);
  \draw (tl) -- (bl);
  \draw (bf) -- (bl);
  \draw (tb) -- (tr);
  \draw (tb) -- (tl);

  \def\tone{.4}\def\ttwo{.75}\def\fone{.36}\def\ftwo{.70}
  \draw ($ (bf)!\tone!(br) $) -- ($ (tf)!\tone!(tr) $) -- ($ (tl)!\tone!(tb) $);
  \draw ($ (bf)!\ttwo!(br) $) -- ($ (tf)!\ttwo!(tr) $) -- ($ (tl)!\ttwo!(tb) $);
  \draw ($ (bf)!\tone!(bl) $) -- ($ (tf)!\tone!(tl) $) -- ($ (tr)!\tone!(tb) $);
  \draw ($ (bf)!\ttwo!(bl) $) -- ($ (tf)!\ttwo!(tl) $) -- ($ (tr)!\ttwo!(tb) $);
  \draw ($ (tl)!\fone!(bl) $) -- ($ (tf)!\fone!(bf) $) -- ($ (tr)!\fone!(br) $);
  \draw ($ (tl)!\ftwo!(bl) $) -- ($ (tf)!\ftwo!(bf) $) -- ($ (tr)!\ftwo!(br) $);

  \definecolor{mygreen}{rgb}{0,0.5,0.13}

  \draw [line width=0.08cm, red] (tf) -- (bf);
  \draw [line width=0.08cm, mygreen] (tf) -- (tr);
  \draw [line width=0.08cm, blue] ($ (tf)!\ttwo!(tr) $) -- ($ (tl)!\ttwo!(tb) $);

  \def\tonenew{.3}\def\ttwonew{.66}\def\tyet{0.8}
	\shade[ball color=red!100, opacity=0.8] ($ (bf)!\tonenew!(tf) $) circle (0.1cm);
        \fill[red!100] ($ (bf)!\tonenew!(tf) $) circle (0.1cm);

        \shade[ball color=blue!100, opacity=0.3] (0.85,0.69) circle (0.1cm);
        \fill[blue!100] (0.85,0.69) circle (0.1cm);
 \node at (0,-3.2) {\tiny $(0,0,0)$};
\end{tikzpicture}
\end{wrapfigure}
\noindent\textit{Illustration of the encoding process.}\
%
Refer to the cube on the right.
%
  The~$x,y,z$ axes are represented by the vertical dimension, 
  and the two
  horizontal dimensions, respectively. 
  The given non-root is the red dot~$\overline{a}
  = (1,0,0)$. The root to encode is the blue
  dot~$\overline{c} = (3,2,1)$. The process starts at the red line~$(t,0,0)$ 
  through~$\overline{a}$, parallel to the~$x$ axis. Replacing 
  the first component~$a_1$ by~$c_1$,
  we check if~$P(c_1,a_2,a_3) = 0$. If not, we travel 
  along the red line
  for~$c_1-a_1 = 2$ units to reach the green 
  line~$(3,t,0)$, parallel 
  to the~$y$ axis. Next we check if~$P(c_1,c_2,a_3) = 0$.
  If not, we travel along the green 
  line for~$c_2-a_2 = 2$ 
  units to reach the blue 
  line~$(3,2,t)$, parallel to the~$z$ axis. 
  We test now if~$P(c_1,c_2,c_3) = 0$, which checks,
  because~$\vc$ is a root.
  The journey ends here. In this example
  it took us~$k=3$ steps to reach
  the root. Note that this was the first~$k$ 
  in~$1,\ldots,n$ that caused~$P$ to vanish on the 
  \emph{hybrid}~$(c_1,\ldots,c_{k-1},c_k,a_{k+1},\ldots,a_n)$.
  
  The blue line, call it $L$, is the one we use to encode the
  root $\vc$. To encode it, we use
  the index of~$c_k$ as a root of the univariate polynomial $P_L(t)$ obtained
  by restricting~$P$ to this line. If this index is $i$ in~$1,\ldots,d$, 
  then we encode the root~$\overline{c}$ by~$(i,k,(c_1,\ldots,c_{k-
  1},c_{k+1},\ldots,c_n))$. In the example, if the index~$i$ 
  happens to be~$1$, then we use~$(1,3,(3,2))$. Note that
  this encoding depends on the given non-root~$\overline{a}$, but any
  non-root serves the purpose of encoding all roots in~$S^n$.
\endgroup

\subsection{Applications}\label{sec:intro:applications}

\subsubsection{\SZ\ Lemma in the Theory}
\label{sec:The formalisation of SZ}
We begin with a short  informal description of the theories, language  and axioms 
we shall use.

\begin{description}

\item[\PV:] This is the language 
with a function symbol for every polynomial-time function, with its meaning specified
by the equations that define it via Cobham's bounded recursion on notation.

\item[\Sonetwo:] The first level of Buss' family of theories \cite{Bus86} for basic number theory whose definable functions are precisely the polynomial-time functions. 
It contains basic axioms for properties of 
numbers (e.g., associativity of product), together with a polynomial induction axiom for \NP-predicates.
The extension of $\Sonetwo$ 
with all $\PV$-symbols and the Cobham equations as axioms is denoted by~$\Sonetwo(\PV)$.
The theory has the same theorems 
as $\Sonetwo$ in
the base language (see~\cite{Bus86}), and
it is customary to abuse notation and still
call it $\Sonetwo$ instead of the heavier $\Sonetwo(\PV)$.
%

\item[\dWPHP(\PV):] The set of \emph{dual weak pigeonhole principle} axioms $\dWPHP(f)$, for every 
po\-ly\-no\-mial-time function symbol $f\in\PV$. This axiom  states the simple counting principle that a 
function~$f$ cannot map surjectively a domain of size~$N$ to a range of size~$2N$ or more; 
namely, there is a point in the set of size~$2N$ that is not covered by~$f$. 
Wilkie (unpublished; see \cite[Theorem 7.3.7]{Kra95})  observed the connection between this principle and 
randomness in computation (within bounded arithmetic): roughly speaking, when~$f$ has a small domain but 
much larger co-domain, with high probability a point in the co-domain will not be covered by~$f$. 
Hence, the ability to  pick such a  point is akin to witnessing this  probabilistic 
argument.     

\item[\Sonetwo+\dWPHP(\PV):] \Sonetwo~(indeed $\Sonetwo(\PV)$; see above), augmented with the axioms $\dWPHP(f)$ for every polynomial-time function symbol $f\in\PV$. 
 This is a theory that can serve as a basis for probabilistic reasoning 
 (close to Je{\v{r}}{\'a}bek's  theory for approximate counting~\cite{Jer07}; cf.~\cite{Jer04}).

\end{description}

\noindent With this notation we can now state the form of \SZ\ Lemma that we prove. Here, and in the rest of this introduction,
let $[q]$ denote the set $\{1,\ldots,q\}$.

\begin{theorem}[\SZ\ Lemma in \Sonetwo; informal, see \Cref{lem:encoding-roots}]\label{thm:intro:SZ in theory} 
Let $P$ be a polynomial of degree $d$, given as an algebraic circuit, 
with integer coefficients and $n$ variables. Then, either $P$ is zero everywhere on \Z, or for every positive integer $q$ there is a (polynomial-time) function $f$ that given any non-root~$\av = (a_1,\ldots,a_n) \in \formalZ^n$ with~$P(\av) \not= 0$,  
returns a function~$f(\av):{\sf codes}\to{\sf roots}$ that maps the set
of codes \emph{surjectively onto} the set of roots in the cube $[q]^n$, and $|{\sf  codes}|\le d\cd n\cd q^{n-1}$.
\end{theorem}


 

Here the codes and the function $f$ are defined according to  the 
encoding scheme 
in Section~\ref{sec:intro:new proof of SZ}. 
Since given a non-root $\va$ the function $f(\va)$ is (provably) 
surjective onto the roots in the cube $[q]^n$, 
the number of roots of $P$ in the cube  is at most 
the number of codes, or~$|{\sf  
codes}|\le d\cd n\cd q^{n-1}$. To actually reason in the theory 
about the size of exponential-size sets like ${\sf  codes}$ we 
could have chosen to invoke approximate counting 
(based on Je{\v{r}}{\'a}bek's theories \cite{Jer07}, which would 
require the inclusion of the dual weak pigeonhole principle \dWPHP). 
However, we opt not to do this for the \SZ\ Lemma. Rather, we will show 
that this formulation of the \SZ\ Lemma, together with the \dWPHP(\PV)\ 
axiom, suffices to apply the lemma in its standard applications, such 
as PIT and finding small hitting sets (see below).

%

\bigskip 
En route to 
the proof of 
\Cref{thm:intro:SZ in theory}, we prove in the theory $\Sonetwo$ one half of the 
\emph{Fundamental Theorem of Algebra} (FTA). 
This is the fundamental theorem of univariate polynomials stating that
every non-zero polynomial of degree $d$ with complex coefficients
has exactly $d$ complex roots.
The theorem naturally splits
into two halves: the half that states that there are \emph{at least~$d$
  roots}, and the half that states that there are \emph{at most~$d$
  roots}. While the \emph{at least} statement relies on special properties of the
complex numbers, the proof of the \emph{at most} statement relies only on the fact
that univariate polynomials over a field admit Euclidean division.
In particular, the \emph{at most} statement 
holds also for polynomials over any
subring of a field; e.g., the integers.


\begin{theorem}[Half of Fundamental Theorem of Algebra in \Sonetwo\, informal; see~Lemma~\ref{lem:univariate}]
Every non-zero univariate polynomial of degree~$d$ with integer coefficients has
at most~$d$ roots on (every finite subset of)~$\Z$.
\end{theorem}

While
the underlying idea of this proof is standard, it 
is somewhat delicate to carry out the argument in $\Sonetwo$ because we need to keep 
track of the bit-complexity of the coefficients that appear along the way in the 
computations.
It is well-known that certain widely-used
algorithms working with integers or rational numbers,
including Gaussian Elimination,
could incur exponential blow-ups
in bit-complexity if careless choices were made in
their implementation; cf.~\cite{GLS1988}.
We also note that Je{\v{r}}{\'a}bek \cite{Jerabekthesis} 
formalised Gaussian Elimination over rationals in $\Sonetwo$,
and proved also the same half of the FTA that we prove, 
but only for finite fields, where exponential blow ups cannot
occur.

\subsubsection{Existence of Hitting Sets and PIT in the Theory}
\label{sec:intro:Formalisation of PIT/HS}

For a field \F\ and a set of  algebraic circuits $\scriptC$ over $\F$ with $n$ 
variables, we say that a set $H \subseteq \F^n$ is a \emph{hitting set for} 
$\scriptC$ if for every non-zero polynomial~$P$ in~$\scriptC$  there exists a point 
$\va \in H$ such that~$P(\va)\neq 0$. In other words, if $P$ is non-zero, $H$ `hits' 
it. 
Hitting sets are important because when they are explicit and small they allow for 
derandomization of~PIT: running through the full hitting set one can test if a 
given algebraic circuit is the zero polynomial or not. 

By \Cref{thm:intro:SZ in theory}, the theory~$\Sonetwo$ proves (by means of 
a surjective map) that every non-zero~$n$-variable algebraic circuit with 
small degree $d$ has relatively few roots in~$[q]^n$. By a counting
argument (or the union bound, cf.~\cite[Theorem 4.1]{SY10}), it follows that for any given bounds~$d$ and~$2^m$ 
on the degree and the number of circuits in the class $\scriptC$, 
there is a set~$H \subseteq [q]^n$ of~$\poly(n,d,m)$ points, 
with~$q = \poly(n,d)$, that intersects the set 
of non-roots of every non-zero circuit in the class. 
This set~$H$ is thus a hitting set of polynomial size, 
and we say it is a hitting set for $\scriptC$ \emph{over $[q]$}. 
We show that this counting argument can now be formalised in the theory~$\Sonetwo+\dWPHP(\PV)$.

\begin{theorem}[Small Hitting Sets Exist in~$\Sonetwo+\dWPHP(\PV)$; informal, 
see \Cref{lem:hittingset}] \label{lem:intro:hittingset} 
For every class $\scriptC$ of algebraic circuits with integer coefficients that is definable in the theory, 
with~$n$ variables, polynomial degree, and polynomial size, 
there exists a polynomial-size hitting set for~$\scriptC$ over~$[q]$ with $q=\poly(n)$. 
\end{theorem}

The argument in the theory makes two uses of the axiom~$\dWPHP$ and
is roughly as follows. 
We begin by showing that, if $q$ is sufficiently large but polynomial, then
a non-zero polynomial with~$n$ variables and polynomial degree always has 
non-root~$\va$ in~$[q]^n$. 
To see this, recall the function $f(\av):{\sf codes}\to{\sf roots}$ 
from \Cref{thm:intro:SZ in theory}, that given a polynomial~$P$ and a  non-root~$\va$ of~$P$, 
surjectively maps all codes of roots to the roots of $P$. 
Note that the set~${\sf roots}$ is a subset of~$[q]^n$, which has size~$q^n$, and
recall that the set~${\sf codes}$ has size at most~$d \cd n \cd q^{n-1}$, where~$d$
is the degree of~$P$. Thus, 
when~$q \geq 2dn$, the~$\dWPHP$ axiom applies to show that there exists
a point $\va_0$ in $[q]^n$ that is not in the range~${\sf root}$ of~$f(\va)$. 
This~$\va_0$ is
thus a non-root of $P$ in the set $[q]^n$, like we wanted.

Next we show how to use this fact to get a hitting set with a second application 
of the~$\dWPHP$ axiom.
Let $\scriptC$ be a class of algebraic circuits with~$n$ variables, 
syntactic degree at most~$d$, and size at most~$s$. 
Consider the function
\begin{equation}
\begin{array}{llllll}
g & : & \scriptC \times [q]^n\times \textsf{codes}^r & \to & \textsf{roots}^r 
 \\
& & (P,\av,\cv_1,\ldots,\cv_r) & \mapsto &
(f(\av)(\cv_1),
\ldots,f(\av)(\cv_r)), 
\end{array}
\label{eqn:intro:defofg}
\end{equation}
where~$\cv_1,\ldots,\cv_r$ are candidate
codes, each from the code-set \textsf{codes} of size~$n \cd d \cd q^{n-1}$, 
and~$\va$ is a potential non-root of~$P$ in $[q]^n$.
The parameter~$r$ should be sufficiently big, but polynomial. 
Then, it follows by construction and the fact proved in the previous paragraph
that \emph{any point outside the range 
of~$g$ is a hitting set}, since it will have a non-root for every algebraic 
circuit in~$\scriptC$. To find the point outside the range of~$g$ 
we invoke the~\dWPHP\ axiom, using again the assumption that $q \geq 2dn$. 

One immediate consequence of Theorem~\ref{lem:intro:hittingset} is that
the theory $\Sonetwo+\dWPHP(\PV)$ proves that the problem of verifying polynomial
identities~$\PIT$ can be solved by polynomial-size Boolean circuits, so 
is in $\Ppoly$. We read this as adding evidence to the claim
that the theory is sufficiently powerful to prove most contemporary results 
in complexity theory.
In particular, it adds interest to the question of proving
that the major lower bound conjectures of computational complexity are consistent 
with~$\Sonetwo+\dWPHP(\PV)$ and stronger theories; see~\cite{KrajicekOliveira2016,CarmosinoKabanetsKolokolovaOliveira2021,AtseriasBussMueller2023}
for more on this line of work.

\subsubsection{Contribution to Reverse Mathematics of Complexity Theory}
\label{sec:intro:Contribution to Reverse math project}
%

The fact that the dual weak pigeonhole principle suffices to prove the
existence of small hitting sets raises a natural question: Is it also
necessary?  A positive answer would provide a \emph{combinatorial}
characterization of the \emph{algebraic} statement that small hitting
sets exist. It would also shed light on the role or the necessity 
of the probabilistic method in the usual proof 
of this existential statement. We show
how to achieve a version of these two goals.

We define a formal scheme of \emph{hitting sets axioms} called~$\HS(\PV)$.
We follow two provisos.
First, in view of the generality
of Theorem~\ref{lem:intro:hittingset}, we define the axiom scheme to 
contain one axiom for each definable class $\scriptC$
of algebraic circuits;
the axiom states
that each slice~$\scriptC_n$, consisting of 
the circuits of $\scriptC$ with~$n$ variables,
has small hitting sets. 
Second, in the definition of the axiom for~$\scriptC$, we need to
decide whether to let it state the existence of 
hitting sets of unspecified but polynomial size, 
or to let it state the existence of hitting sets of some 
specified polynomial size.
The bound established in Theorem~\ref{lem:intro:hittingset} 
is actually of the form~$\poly(m,n)$ where $m$ is the logarithm  
of the number of circuits in the $n$-th slice, and $\poly(m,n)$
refers to a fixed polynomial of $m$ and $n$. This dependence on $m$
is common in most proofs of existence by the union bound. While the
claim that hitting sets of any possibly larger but unspecified
size exist would of course be also true, it turns out that asking the 
axiom to provide a hitting set of some fixed polynomial bound seems 
crucial in the proof 
of necessity of $\dWPHP(\PV)$ that we are after. We chose the latter
because this is what is sufficient, and it is still natural.

These two provisos motivate the following definition (informal; 
see~Definition~\ref{def:hspv}):

\begin{description}
\item[\HS(\PV):] The set of \emph{hitting set} axioms $\HS(g)$, for
  every $g\in\PV$. This states that if $g$ defines a class~$\scriptC$
  with its $n$-th slice having $2^m$ algebraic circuits with $n$ variables, 
  polynomial degree, and polynomial size, then there is a hitting 
  set for $\scriptC$ over $[q]$ of size $\poly(m,n)$, with $q =
  \poly(n)$. 
\end{description}

\noindent With the right definitions in place we can state the theorem that characterizes
the proof-theoretic strength of the existence of small hitting sets:

\begin{theorem}[Reverse Mathematics of Hitting Sets; informal, see Theorem~\ref{thm:equivalence}]
  \label{thm:intro:equivalence}
  The axioms schemes $\dWPHP(\PV)$ and $\HS(\PV)$ are provably equivalent over the theory $\Sonetwo$.
\end{theorem}

As remarked earlier, the sufficiency claim follows from
Theorem~\ref{lem:intro:hittingset}.  To prove the necessity we show
how to use a hitting set to find a point outside the range of any
given polynomial-time function~$f : [N] \to [2N]$. To do this we
design a class~$\scriptC_f$ of~$N$ many low-degree algebraic circuits 
each vanishing on the appropriate representation of a point in the 
image of~$f$. We do so in such a way that a hitting set for~$\scriptC_f$
will correspond to an element in~$[2N]\setminus\Img(f)$,
completing the proof. To make this actually work we need to
use a technique known as \emph{amplification}, which goes
back to the work of Paris-Wilkie-Woods~\cite{PWW88}. 
The same kind of technique was discovered even earlier, in cryptography, 
to build pseudorandom number generators 
from hardcore bits; see the work of Blum-Micali~\cite{BlumMicali1982}.
The details of this argument can be found in
Section~\ref{sec:complexity} and Appendix~\ref{app:normalizationandamplification}.

\subsubsection{Application to Computational Complexity and Range Avoidance Problem}
\label{sec:intro:Applications to computational complexity and Range Avoidance Problem}

The proof-sketch we gave for Theorem~\ref{thm:intro:equivalence} reveals a
two-way connection between the computational problem of finding
hitting sets and the so-called \emph{Range Avoidance Problem}, or~$\AVOID$. 
The latter problem 
asks to find a point outside the range of a given function~$f : [N] \to [2N]$. 
In recent years, $\AVOID$ has
been studied with competing names. It was first studied by
Kleinberg-Korten-Mitropolsky-Papadimitriou~\cite{KKMP21}, calling it $1$-$\EMPTY$, 
and later by Korten~\cite{Korten21}, renaming it~$\EMPTY$. Those works
defined it as the canonical complete problem for a complexity class
of total search problems they called~$\APEPP$. Ren-Santhanam-Wang~\cite{RSW22} 
studied it too, calling it $\AVOID$, 
in their range avoidance problem approach to circuit lower bounds. 

Our new coding-based proof of the existence
of hitting sets shows that its associated search problem is in $\APEPP$. 
The proof of Theorem~\ref{thm:intro:equivalence} 
yields its completeness in the class:

\begin{theorem}[Completeness of Finding Hitting Sets; informal, see Theorem~\ref{thm:intro:completeness}]
\label{thm:intro:completeness}
The total search problem that asks to find witnesses for the hitting set axioms of the scheme~$\HS(\PV)$
is $\APEPP$-complete under $\PtoNP$-reductions. 
\end{theorem}

\noindent 
Perhaps not too surprisingly, the proof of Theorem~\ref{thm:intro:completeness} is almost identical to the 
proof of Theorem~\ref{thm:intro:equivalence}.
Indeed, the necessity for the $\NP$-oracle in the reductions 
comes (only) from the use of the 
amplification technique in the proof, as in earlier
uses of this method; cf.~\cite{Korten21}.

A handful of complete problems for $\APEPP$ were known before, and some required $\PtoNP$-reductions 
too (cf.~\cite{KKMP21,Korten21}). But, to our knowledge, none of these complete problems 
related to the problem of constructing hitting sets for algebraic
circuits. Here we used the new constructive 
proof of the Schwartz-Zippel Lemma to find an example of this type.

\section{Preliminaries}

\subsection{Theories of Bounded Arithmetic}

Here we define the formal theories we work with. Our main results are 
proved in the theories~$\Sonetwo$ and its 
extension~$\Sonetwo+\dWPHP(\PV)$. 
These are single-sort theories for arithmetic, namely,
number theory over \N. The theory $\Sonetwo$ is the first
in Buss' hierarchy of theories of bounded arithmetic,
and defines precisely the polynomial-time computable functions~\cite{Bus86}.
The theory $\Sonetwo+\dWPHP(\PV)$ denotes its extension with
the Dual Weak Pigeonhole Principle (dWPHP) for polynomial-time functions (\PV),
which allows to reason about probabilities and do approximate counting,
based on the work of Je{\v{r}}{\'a}bek\ \cite{Jer07}. 
For a rigorous yet concise survey on bounded arithmetic the reader is referred 
to Buss~\cite{Bus97} (for a full treatment of the area see \cite{CN10,Kra95}; 
as well as a contemporary survey by Oliveira~\cite{Oli24}).  
We refer the reader to Tzameret and Cook \cite{TC21} for treatment of 
algebraic circuits over the integers inside theories of bounded arithmetic. 
Here we use a similar treatment of algebraic circuits, 
discussed in the next section.

\para{The language \LA.} \label{sec:languageLA}
The language \LA\ contains the 
function symbols~$+,\times,\floor{\cd/2},|\cd|,\#, 0, 1$, and
the relation symbol~$\leq$. A finite 
collection of axioms called \BASIC\ gives these
symbols their intended meaning:~$0$ and~$1$ are
constant symbols for the numbers~$0$ and~$1$, the symbols~$+$ 
and~$\times$ denote addition and multiplication of numbers, and
the unary functions~$\floor{\cd/2}$ and~$|\cd|$ denote
\emph{rounded halving} and \emph{binary length}, respectively. 
Precisely, if $x$ is a natural number, then $\floor{x/2}$
is the largest natural number bounded by $x/2$, and $|x|$
is the length of $x$ written in binary notation,
except for~$x=0$ where~$|x|$ is defined as~$0$ by convention.
I.e.,~$|x| = \ceil{\log_2(x+1)}$.
We say that~$|x|$ is the \emph{length of $x$}. 

The only
non-self-explanatory symbol is the ``smash'' binary function
symbol~$\#$ introduced by Buss \cite{Bus86}. The
intention of $x\# y$ is $2$ raised to the power of \emph{the
length} of $x$ times \emph{the length} of $y$, namely
$2^{|x|\cd|y|}$. This function allows us to work with 
strings as follows: assume we wish to talk about a binary
string $S$ of length~$n$ in the
theory. 
Since the theory only talks about numbers, we represent
the string as the number~$x$ between~$2^{n-1}$ and~$2^n-1$
whose unique binary representation of length~$n$ is~$S$. 
Note that the length~$n$ of $S$ is precisely~$|x|$.
To consider 
a polynomial growth rate
in the length of $S$, we need to be able to express strings of
length $n^c$, for a constant $c$. Strings of length at 
most~$n^c$ are
encoded, as before, by numbers of magnitude less than
\begin{equation}
2^{n^c} = 2^{|x|^c} = 
2^{\tiny\underbrace{|x|{\cd}|x| \cdots |x|}\atop c{\text{ times}}} = 
x\#(x\# \cdots \#(x\#x)) \text{ [$c$ times]}. 
\end{equation}
We write~$n\in\Log$ to mean that~$n$ serves 
as the length of some string, namely,~$n=|x|$ for 
some~$x$. Note that in bounded arithmetic theories
where exponentiation may not be total, the formula
$\forall n\exists x\ |x|{=}n$ is not provable.
We use~$\forall n{\in}\Log\ \psi$ and $\exists n{\in}\Log\ \psi$
as shorthand notations for
the formulas~$\forall x\forall n\ (|x|{=}n \rightarrow \psi)$
and~$\exists x\exists n\ (|x|{=}n \wedge \psi)$.

\para{The Theory $\Sonetwo(\PV)$.}

We need to define several polynomial-time computable 
functions in the theory and show
that the theory proves some basic properties of these functions.
We work in the extension of the language $\LA$ with 
the language $\PV$ which has a new symbol for every polynomial-time
computable function (to be precise, it has a new symbol for every function defined
by Cobham's bounded recursion on notation \cite{Cob65}). 
The theory will be Buss's~$\Sonetwo$ extended with the axioms defining the added
$\PV$-symbols (\cite{Bus86}; see \cite{Bus97}). This theory is sometimes
denoted $\Sonetwo(\PV)$ but we will use the shorter notation 
$\Sonetwo$ if this does not lead to confusion as
$\Sonetwo(\PV)$ is a conservative extension
of $\Sonetwo$ for statements in the \LA-language.
For the precise definition of $\Sonetwo(\PV)$ see Kraj\'{i}\v{c}ek \cite[page~78]{Kra95}.

Objects like strings, circuits, etc.~in the theory are coded by 
numbers~(cf.~\cite{Bus86,Kra95}). In this sense, we may refer to an object in the 
theory, meaning formally its code.
The inputs and outputs of $\PV$-functions are natural numbers.
However, the polynomial-time algorithms/Turing machines 
that compute the functions
as~$\PV$-symbols get their inputs presented in binary notation
as binary strings. The statement that $n$ is an argument that is 
presented to a~$\PV$-function in unary notation means that it
is expected that $n \in \Log$ (so $2^n$ exists) and that
the argument is actually~$2^n$, so in this case
the binary representation $100\cdots 0$ of~$2^n$ is given in 
the input of the polynomial-time algorithm that computes the 
function.
An alternative convention would be to think of $\PV$-functions 
as getting inputs of two \emph{sorts}: numbers in unary 
notation and numbers in binary notation.
This is the approach taken in the \emph{two-sorted theories}
of bounded arithmetic~\cite{CN10}.

\para{Definable sets and set-bounded quantification.} \label{sec:definable}
Let $\Phi$ be a class of formulas in the language~$\LA$.
A set of natural numbers $S \subseteq \naturals$ is called $\Phi$-definable (with parameters)
if there exists a formula $\varphi(x; y)$ in $\Phi$, with all free variables indicated,
such that for some $c \in \naturals$ we have
$S = \{ a \in \naturals : \naturals \models \varphi(a; c) \}$.
This is definability in the \emph{meta-theory}, or 
\emph{in the standard model}. In a formal theory~$T$ such as $\Sonetwo$, a 
formula~$\varphi(x;y)$ defines a set in every model $M$ of the theory for every
choice of the parameter $c \in M$ in the model:
\begin{equation}
[\varphi(x; c)]^M := \{ a \in M : M \models \varphi(a; c) \}. \label{eqn:definable}
\end{equation}
When the model or the parameter are not specified, a \emph{definable set} is simply given
by the formula $\varphi(x;y)$ that \emph{defines} it (parametrically in $y$),
and we write $[\varphi(x;y)]$.

For example, if we take $\varphi(x;y) := 1{\leq}x \wedge x{\leq}y$, then 
for every standard $c \in \naturals$ we have
\begin{equation*}
[1{\leq}x\wedge x{\leq}c]^{\naturals} = [c] := \{1,2,\ldots,c\}.
\end{equation*} 
When this is not confusing, we use the same 
notation $[c] := \{1,2,\ldots,c\}$ in the theory instead of the more accurate $[1{\leq}x\wedge x{\leq}c]$, 
even when $c$ is a free variable. Also, we use $x{\in}[c]$ 
as short-hand notation for the formula~$1{\leq}x \wedge x{\leq}c$. More generally, if we
declare that a formula~$\varphi(x;c)$ \emph{defines a set which we denote~$S_c$}, then
we are entitled to use the following short-hand notations:
\begin{equation}
\begin{array}{lllll}
x{\in}S_c & \equiv & \varphi(x;c) \\
\exists x{\in}S_c\ \psi & \equiv & \exists x\ (\varphi(x;c) \wedge \psi) \\
\forall x{\in}S_c\ \psi & \equiv & \forall x\ (\varphi(x;c) \rightarrow \psi).
\end{array}
\label{eqn:shorthand}
\end{equation}
When $S_c$ is a bounded set, i.e., there exists a term $t(c)$ such that
every~$x \in S_c$ satisfies $x \leq t(c)$ provably in the theory, then
the quantifiers in~\eqref{eqn:shorthand} can be bounded: $\exists x{\leq}t(c)\ (\varphi(x;c) \wedge \psi)$ and $\forall x{\leq}t(c)\ (\varphi(x;c) \rightarrow \psi)$.

\para{Dual Weak Pigeonhole Principle.}

  Let $f$ be a $\PV$-symbol. The axiom $\dWPHP(f)$
  is the universal closure of the following formula with free variables $a,b,c$:
  \begin{equation}
  \dWPHP^a_b(f_c) := (b{\geq}2a \wedge a{\geq}1 \rightarrow \exists y{\in}[b]\ \forall x{\in}[a]\ f_c(x)\not=y),
  \end{equation}
  where $f_c(x)$ is notation for $f(x,c)$, 
  thinking of $c$ 
  as a \emph{parameter} for $f$. 
  We write~$\dWPHP(\PV)$ for the axiom-scheme that contains all axioms $\dWPHP(f)$ for
  all $\PV$-symbols $f$.

  The presence of the parameter $c$ in $f_c$ makes the statement of $\dWPHP(f)$ more expressive. 
  For example, the parameter
  may specify the sizes~$a$ and~$b$ of the intended domain and range of a function $f_c : [a] \to [b]$,
  say as $c = \langle a,b\rangle$, for an appropriate pairing function $\langle \cdot,\cdot \rangle$. 
  Another example of the power of parameters is the following.
  The class of $\PV$-functions admits a natural \emph{universal-like} 
  function, that we call $\evalboolean$, which provably in $\Sonetwo$ evaluates
  any $\PV$ function $f$: the theory
  $\Sonetwo$ proves $\forall b\ \forall x{<}b\ \eval(C_f(b),x){=}f(x)$ for a certain 
  natural $\PV$-function $C_f$.
  It follows from this that
  the infinite axiom scheme $\dWPHP(\PV)$ is equivalent, over $\Sonetwo$, to 
  the single axiom $\dWPHP(\eval)$.

  \subsection{Polynomials and Algebraic Circuits}

Let $\mathbb{G}$ be a ring.
Denote by~$\mathbb{G}[\overline x]$ the ring of (commutative)
polynomials with coefficients from~$\mathbb{G}$ and variables (indeterminates)~$\overline x := \{x_1, x_2,\dots\}$. A polynomial is a 
formal linear combination of monomials, whereas a monomial is a product of 
variables. Two polynomials are identical if all their monomials have the same 
coefficients. The (total) degree of a monomial is the sum of all the powers of 
variables in it. The (total) degree of a polynomial is the maximum (total)
degree of a monomial in it. The degree of an individual variable in a monomial 
is its power, and in a polynomial it is its maximum degree 
in the monomials of the polynomial. 
The maximum individual degree of a polynomial is the maximum degree of
its variables.

\paragraph{Algebraic circuits and formulas.}
Algebraic circuits and formulas over the ring $\mathbb{G}$ compute polynomials 
in $\mathbb{G}[\overline x]$ via addition and
multiplication gates, starting from the input variables and constants from the 
ring. In the rest of this paper the ring $\mathbb{G}$ will be fixed
to the integers $\Z$.

An \emph{algebraic circuit} (with parameters) is a directed acyclic graph (DAG) with its nodes labelled.
The sources of the DAG are labelled by the name of an input. 
The internal nodes are labelled
by gates of types $+$ and $\times$. The inputs of the circuit are split into variables and parameters; 
a parameter-free circuit is one without parameter
inputs. The circuit may come with an implicit integer parameter assignment,
in which case the corresponding sources of the DAG are labelled by the corresponding integer. 
We say that the constants are \emph{plugged} into the circuit.
An \emph{algebraic formula} is an algebraic circuit whose DAG is a tree.

\paragraph{Size and syntactic degree.}
The representation size 
  of a circuit is the number of
  bits that are needed to represent the DAG, the
  operation of each internal node of the graph, and the names of the 
  variables and the parameters.
  If the circuit
  comes with implicit integer parameter assignment, then the representation 
  size also includes the 
  number of bits in their binary representations. 

An algebraic circuit $C$ can also be understood as a \emph{straight-line program},
which is a finite sequence $C_0,C_1,\ldots,C_t$ of labelled 
\emph{gates}. Each gate $C_i$
is of one of four types: (1) a variable input gate, which is labelled by 
the name of a variable, or (2) a parameter input gate, which is labelled by 
the name of a parameter and perhaps also with an integer constant 
if the circuit comes with an implicit parameter assignment, 
or (3) an addition gate, which is labelled by the symbol $+$ and 
two integers $j < i$ and $k < i$ that specify the two 
operands in the addition $C_j + C_k$ that is computed at the gate, 
or (4) a multiplication gate, which is labelled by the 
symbol $\times$ and two integers $j < i$ and $k < i$ that specify the two 
operands in the multiplication $C_j \times C_k$ 
that is computed at the gate. The last gate $C_t$ is the \emph{output} of 
the circuit.

The \emph{syntactic degree} $d_i$ 
of the gate~$C_i$ is defined inductively on $i$: if $C_i$ is a variable
gate or a parameter gate, then $d_i = 1$; if $C_i$ is an 
addition gate $C_j + C_k$, then $d_i = \max\{d_j,d_k\}$; and if $C_i$ is a multiplication gate $C_j \times C_k$, then 
$d_i = d_j + d_k$. The \emph{syntactic degree of $C$} is the syntactic degree of its output gate $C_t$. By induction on 
the number of gates in the circuit, it is easy to see 
that the syntactic degree of a circuit with $t$ gates is at most 
$2^t$. For formulas, a better upper bound is $t$. 

A refinement of the concept of syntactic degree is \emph{syntactic individual degree}.
Suppose that~$C$, with straight-line program $C_0,C_1,\ldots,C_t$, has $n$ variables and $m$ parameters.
Let~$u \in \{1,\ldots,n+m\}$ be the name of a variable or a parameter. 
The \emph{syntactic individual degree of gate $C_i$ on $u$}, denoted by $d_{i,u}$, is defined also inductively 
on $i$: If $C_i$ 
is a variable or parameter gate with label $u$, 
then $d_{i,u}=1$; if $C_i$ is a variable or a parameter gate with 
label~$v \not= u$, then $d_{i,u} = 0$; if $C_i$ is an addition gate $C_j + C_k$, 
then~$d_{i,u} = \max\{ d_{j,u}, d_{k,u} \}$; 
if $C_i$ is a multiplication gate $C_j \times C_k$, then $d_{i,u} = d_{j,u} + d_{k,u}$. 
The \emph{syntactic individual degree of $C$ on~$u$} is $d_{t,u}$. Clearly, the syntactic degree 
of~$C$ is bounded by the sum of the individual degrees, 
and each individual degree is bounded by the syntactic degree. What we call syntactic degree 
is sometimes called syntactic \emph{total} 
degree.

\section{\SZ\ Lemma in the Theory}

\subsection{Notation and formalisations}

Here we discuss formalisation in the theory and fix some notation.
For every natural number~$n \in \mathbb{N}$ we write $[n] = \{1,\ldots,n\}$.
We fix some standard and efficient encoding for pairs, tuples, and lists of
natural numbers as natural numbers, with its usual properties provable in 
the theory (cf.~\cite{Bus86}). We discussed already in Section~\ref{sec:languageLA}
how binary strings
are encoded in the theory. For $n \in \Log$, we write~$\{0,1\}^n$ for the 
definable set of strings of length $n$, with~$n$ as a parameter.

\paragraph{Ring of integers in the theory.}
Integers are encoded in the theory in the sign-magnitude
representation as pairs~$(b,m)$ where~$b \in \{0,1\}$
is the sign, and~$m$ is the magnitude.
The intention is that the pair~$(b,m) \in \{0,1\} \times \mathbb{N}$ 
encodes the integer $(-1)^b \cdot m$. Note that $0$ has two codes.
The \emph{bit-complexity} of the integer is the length~$|m|$ of 
its magnitude.
When fed as argument into an algorithm operating with strings,
the integer represented by $(b,m)$ is presented as the pair~$(b,m)$ itself,
with $m$ written in binary notation. The set of codes~$(b,m)$ of integers
in the theory is definable by the quantifier-free formula~($b{=}0\vee b{=}1$), 
and is denoted by $\formalZ$. The addition, subtraction, and 
multiplication operations on $\formalZ$ are $\PV$-functions. The 
ordering relation~$\leq$ on $\formalZ$ is a $\PV$-predicate.
The basic properties of these symbols are provable in the usual theories.

\paragraph{Algebraic circuits in the theory.}
Algebraic circuits and formulas over \Z\ are formalised in the theory as 
labelled graphs with the integer constants
that may appear on its 
leaves encoded in the sign-magnitude representation discussed
above. We refer the reader to \cite[Section~3.1.1]{TC21} for more details on 
encoding algebraic circuits over \Z\ in the theory. 
We define some $\PV$-functions that operate with codes of 
algebraic circuits:

\begin{itemize} \itemsep=0pt
\item $\size(C)$: the~$\PV$-function that computes the
number of gates of the algebraic circuit~$C$.
In any reasonable explicit encoding of graphs we have $\size(C) \leq |C|$.
\item $\dimension(C)$: the~$\PV$-function that computes
the number of variables, or indeterminates, of the algebraic 
circuit~$C$, called the \emph{dimension}
  of~$C$.
\item $\parametricdimension(C)$: the~$\PV$-function that computes
the number of parameters of the algebraic circuit $C$, 
called the \emph{parametric dimension} of $C$. 
\item $\totaldegree(C)$: the~$\PV$-function that computes the 
\emph{syntactic} total degree of the algebraic circuit~$C$. Recall
that this is bounded by $2^{\size(C)}$, hence by $2^{|C|}$,
so its binary representation fits in $|C|$ bits. Recall
also that our definition of syntactic degree counts the parameter inputs as
contributing to the degree. One consequence of this is that
if the circuit is fed with integers for its $n$ variables, and 
integers for its $m$ parameters, 
all of bit-complexity at most $s$, then the output has bit-complexity 
polynomial in~$d,s,n,m$, where $d$ is
the total degree.
\item $\individualdegree(C,i)$: 
the $\PV$-function that computes the \emph{syntactic}
individual degree of the~$i$-th input in the algebraic circuit $C$, 
where $1 \leq i \leq n+m$, and $n$ and $m$ are the dimension and the parametric 
dimension of $C$. The \emph{maximum syntactic individual degree} of the circuit
is denoted $\maxindividualdegree(C)$. 
When $C$ has a single variable, we
write $\degree(C)$ instead of $\maxindividualdegree(C)$.
\item $\evalarithmetic(C,\av,\pv,d)$: the~$\PV$-function for the
  standard polynomial-time algorithm which, given
  an integer $d$ in unary notation, given an
  algebraic circuit~$C$ of syntactic maximum individual degree~$d$, 
  given vectors of integer inputs~$\av =  (a_1,\ldots,a_n)$ for the variables,
  and~$\pv = (p_1,\ldots,p_m)$ for the parameters, 
  evaluates $C$ on inputs~$\av$ and~$\pv$, with the 
  gates that are labelled by~$+$ 
and~$\times$ interpreted as the addition
and multiplication operations of the integers. This is defined inductively 
on the structure of~$C$
and the standard algorithm runs in time polynomial in the size of its input 
because the syntactic degree $d$ is given in unary notation. Here
we use also the fact that the definition of syntactic degree takes the parameter inputs into account.
In other words, this $\PV$-function models the evaluation of algebraic circuits of
polynomial-degree with constants of polynomial bit-complexity.
See \cite[Sec.~11.1]{TC21} how to define this function in~\PV~(first balancing 
the algebraic circuit of polynomial syntactic-degree, and then evaluating the 
balanced circuit; in \PV~however, the initial balancing step is not necessary). 
If the parameter inputs~$\pv$ are plugged into $C$, then
  the notation~$C(\av)$ abbreviates~$\evalarithmetic(C,\av,\pv,\totaldegree(C))$.
\end{itemize}

We need also some notation for \emph{definable sets of circuits}; see~Section~\ref{sec:definable} for a discussion on definability in the theory.
With the encodings discussed so far, 
all the sets below are quantifier-free definable in the language $\LA$.

\begin{itemize}
\item $\Circuits$: the set of (codes of) algebraic circuits.
The subset of formulas is denoted~$\Formulas$.
\item $\Circuits(n,d)$: the set of (codes
  of) algebraic circuits with at most~$n$ indeterminates and syntactic maximum 
  individual degree at most~$d$. The subset of formulas is denoted~$\Formulas(n,d)$. 
\item $\Circuits(n,d,s)$: the subset of $\Circuits(n,d)$
  whose elements have representation size at most~$s$. The subset of formulas
  is denoted $\Formulas(n,d,s)$.
\item $\UniPoly(d)$: the set of (codes of)
  univariate polynomials with integer coefficients and degree at most~$d$, 
  written (as formulas) in explicit sum of monomials form. In
  other words, if~$\formvar$ denotes the indeterminate,
  then~$\UniPoly(d)$ denotes the set of polynomials of the
  form~$P(\formvar) = c_0 + c_1 \formvar + c_2 \formvar^2 + \cdots +
  c_d \formvar^d$, where, for~$i = 0,\ldots,d$, the~$i$-th
  coefficient~$c_i$ is an integer. Note that when we say \emph{degree at
  most~$d$} we do \emph{not} require that~$c_d \not= 0$; that would
  be degree \emph{exactly}~$d$. 
  \item $\UniPoly(d,s)$: the subset of $\UniPoly(d)$ in which all 
  coefficients are integers of bit
  complexity at most~$s$; i.e., each coefficient~$c_i$ is an integer
  in the interval~$[-2^s+1,2^s-1]$.
  We write $\coef(P,i)$ for the $\PV$-function that extracts the coefficient
of the term of degree $i \in \{0,\ldots,d\}$ of the univariate polynomial 
$P \in \UniPoly(d)$. This is trivially computable in polynomial time
because the polynomials in $\UniPoly(d)$ are given in explicit sum of
monomials form. 
\end{itemize}

Finally, we need to define the concept of \emph{semantic equivalence over a definable set}. If~$F$ and $G$ denote algebraic circuits with the same number $n$ of indeterminates, and $S$ is a (finite or infinite)
definable set of integers, with its membership predicate $x \in S$ definable by a formula (with or without
parameters), 
then the notation $F \equiv_S G$ stands
for the (possibly unbounded) formula
\begin{equation}
    \forall \av{\in}S^n\ F(\av){=}G(\av).
\end{equation}
In words, this formula says that evaluations
of $F$ and $G$ agree on every $n$-vector of integers~$\av = (a_1,\ldots,a_n)$ 
in $S^n$.
If $q$ is an integer (in the theory), then we use $S_q$ to denote the definable
set $\{0,1,\ldots,q-1\}$ of integers between $0$ and $q-1$; i.e.,
\begin{equation}
    S_q := \{0,1,\ldots,q-1\} \subseteq \formalZ
    \label{eqn:Sq}
\end{equation}
We write $F \equiv_q G$ instead of $F \equiv_{S_q} G$,
and the corresponding formula can be written with 
bounded quantifiers:~$\forall \av{\leq}q\ (\av{\in}S_q^n \rightarrow F(\av){=}G(\av))$ 
or, more precisely, 
the bounded quantifier is~$\forall \av{\leq}t(n,q)$, where $t(n,q)$ is
the $\PV$-function that bounds the
encodings of the elements of~$S_q^n$.
In the other extreme case in which $S$ is the set of \emph{all} 
integers $\formalZ$,
we write $F \equiv S$, and the corresponding formula has unbounded 
quantifiers~$\forall \av{\in}\formalZ^n\ F(\av){=}G(\av)$.

\para{Definable classes of algebraic circuits.}
We introduce next the concept of a \emph{definable class of algebraic circuits}.
The intention is to capture the usual practice in algebraic circuit complexity
of considering different subclasses of circuits with varying parameters (cf.~\cite{SY10}).
Examples include: formulas, syntactic multilinear circuits,
uniform families, projections of the determinant polynomial, 
arithmetizations of Boolean circuits, algebraic branching programs,~etc. 
We have already
seen two classes of definable classes: the class of circuits~$\Circuits$, and 
the class of formulas~$\Formulas$. 
The following definition generalizes these:

\begin{definition}[In~$\Sonetwo$]
\label{def:definable classes of algebraic circuits}
  A \emph{definable class of algebraic circuits} is a subset
  $\scriptC \subseteq \Circuits$ of (codes of) algebraic circuits
  that comes with a $\PV$-function~$g$, the \emph{decoding function}, 
  that is surjective onto~$\scriptC$.
  The polynomial-time algorithm that computes~$g$ as a $\PV$-function
  must satisfy the following condition: given~$n,d,s$ in unary 
  and given a string~$x \in \{0,1\}^m$,
  the function $g_e(n,d,s,x)$ outputs an algebraic 
  circuit in~$\Circuits(n,d,s)$,
  for all settings of the parameter~$e$.
\end{definition}

For the rest of this section, 
let $\scriptC$ be a definable class of algebraic circuits
with decoding function~$g$. If $C$ is an algebraic circuit in~$\scriptC$
and~$g_e(n,d,s,x) = C$, then we say that~\emph{$x$ is a description of $C$ as a
member of~$\scriptC$}.  The
\emph{description size of~$C$ as a member of $\scriptC$}
is the length as a string of its smallest description as a member 
of~$\scriptC$. 

We write $\scriptC_e(n,d,s,m)$ for the \emph{slice} of algebraic circuits
in~$\scriptC$ of the form~$g_e(n,d,s,x)$ with~$x \in \{0,1\}^m$; in symbols:
\begin{equation*}
\scriptC_e(n,d,s,m) := \{ g_e(n,d,s,x) : x \in \{0,1\}^m \}.
\end{equation*}
Note
that $\scriptC_e(n,d,s,m)$ is always a subset of $\Circuits(n,d,s)$, but
its cardinality is at most $2^m$, which may be much smaller than the
cardinality of~$\Circuits(n,d,s)$. When $2^m \ll 2^s$, we say
that the class is~\emph{sparse}. An extreme case of this occurs
when $\scriptC$ is a class of algebraic circuits determined
by their dimension, say $(C_n)_{n \geq 1}$ with $C_n \in \Circuits(n,d(n),s(n))$,
where each $\scriptC(n,d,s)$ has at most one member $C_n$.
A typical example is the class of \emph{determinant polynomials}
$(\det_{n})_{n \geq 1}$, where $\det_n$ denotes the polynomial with~$n^2$ 
indeterminates that computes the determinant of the $n \times n$ matrix
given by its~$n^2$ inputs. 
A related but less extreme example of a sparse definable class
is the class of determinants of Tutte matrices of graphs.
In this example, the representation size of a member in this class 
is determined by the number of edges of the underlying graph.

\subsection{Fundamental Theorem of Algebra: the Univariate Case}

The Fundamental Theorem of Algebra (FTA) states a fundamental
property of univariate polynomials over the field of complex numbers.

\begin{theorem}[Fundamental Theorem of Algebra]
Every non-zero univariate
polynomial of degree~$d$ with coefficients in the field
of complex numbers has
exactly~$d$ roots in the complex numbers.
\end{theorem}

The theorem naturally splits
into two halves: the half that states that there are \emph{at least~$d$
  roots}, and the half that states that there are \emph{at most~$d$
  roots}. While the \emph{at least} statement relies on special properties of the
complex numbers, the proof of the \emph{at most} statement relies only on the fact
that the class of univariate polynomials with coefficients in a field
forms a Euclidean domain where the Euclidean division algorithm
applies. In particular, this means that the \emph{at most} statement 
holds also for polynomials over any
subring of a field; for example, polynomials over the ring of integers.

In this section we show that the theory~$\Sonetwo$ is able to formalise
the standard proof of the second half of the FTA for the class of
univariate polynomials over the ring of integers. This will be used as a building 
block in the next section. While
the underlying idea of this proof is completely standard, it 
is somewhat delicate to carry out the argument 
in $\Sonetwo$ because 
we need to keep control of the bit-complexity of the 
coefficients that appear along the way in the computations.

Let~$d$ and $q$ be small non-negative integers. Let~$P$ be an element
of~$\UniPoly(d)$, that is,~$P$ is (the code of) a univariate
polynomial with integer coefficients~$c_0,c_1,\ldots,c_d$ of arbitrary 
bit-complexity, and degree at most~$d$. 
If we write~$\formvar$ for
the indeterminate, then
\begin{equation}
P(\formvar) = c_0 + c_1 \formvar + \cdots +
c_d \formvar^d.
\label{eqn:givenP}
\end{equation}
Define
\begin{equation*}
\begin{array}{lcl}
S_q & := & \{0,\ldots,q-1\} \subseteq \Z,  \\
Z_{P,q} & := & \{ u \in S_q : P(u) = 0 \}. 
\end{array}
\end{equation*}
Note that~$Z_{P,q}$ is the set of roots of~$P(\formvar)$ in the
set~$S_q$. The goal is to show that if $P$ is not the
zero polynomial, then $Z_{P,q}$ has cardinality
at most $d$. For later reference, 
we state this upper bound in the form of the existence of 
a surjection from the set $[d]$ onto the set $Z_{P,q}$. 

\begin{remark}
With $P$ given 
as in~\eqref{eqn:givenP}, where the number of roots is polynomial 
in the representation size of $P$,
we could also choose to state the upper bound by listing
the roots. Indeed, provided~$d$ is polynomial in the representation size,
this holds true even if $P$ is given by a univariate algebraic circuit,
as we will later see.
However, in the multivariate case with $n > 1$ variables, 
the number of roots in~$S_q^n$
may be exponential in the representation size of $P$, even if $d$ is small,
so there we need to resort to the surjective mapping formulation of 
the upper bound. For this reason and to be able to reuse it,
we state the upper bound for the univariate case with a 
surjective mapping too.
\end{remark}

We define a $\PV$-function
\begin{equation*}
h_{P,q} : [d] \to S_q \cup \{q\}
\end{equation*} 
by the polynomial-time algorithm that computes it.  The
input to the algorithm for~$h_{P,q}(i)$ is the triple~$(P,q,i)$ with~$P$
as above, $q$ given in unary notation, and~$i \in [d]$.
In the algorithm, let~$u$ loop over
the set~$S_q$ and evaluate~$P(\formvar)$ at~$\formvar=u$, keeping a
counter of the number of distinct roots found along the
way. If~$P(\formvar)$ has at least~$i$ distinct roots in~$S_q$,
then~$h_{P,q}(i)$ is defined to be the~$i$-th smallest root
of~$P(\formvar)$ in~$S_q$. Otherwise,~$h_{P,q}(i)$ is set to the
value~$q$ seen as an end-of-list marker. In other words, the
sequence~$(h_{P,q}(1),\ldots,h_{P,q}(d))$
equals~$(r_1,\ldots,r_t,q,\ldots,q)$ where~$r_1 < \cdots < r_t$ is the
ordered list of the first~$d$ roots of~$P(\formvar)$ in~$S_q$. The next
lemma states that~$\Sonetwo$ proves that if~$P(\formvar)$ does not
vanish everywhere on~$S_q$, then this enumeration indeed covers all
the roots of~$P(\formvar)$ in~$S_q$.

\begin{lemma}[Second Half of Fundamental Theorem of Algebra in~$\Sonetwo$] \label{lem:univariate} 
For all~$d,q \in \Log$ and every degree-$d$ polynomial~$P \in \UniPoly(d)$,
  if not all coefficients of~$P$ are zero, then 
  the sequence~$h_{P,q}(1),h_{P,q}(2),\ldots,h_{P,q}(d)$ contains all the roots
  of~$P$ in $S_q$; i.e., the function~$h_{P,q}$ is surjective
  onto~$Z_{P,q}$.
\end{lemma}

\begin{proof}
  Fix~$d,q \in \Log$ and let~$P$ be an element
  of~$\UniPoly(d)$. Let $s$ be the bit-complexity of 
  the coefficients of $P$, so $s$ is a length and~$P \in \UniPoly(d,s)$.
  The plan for the proof is to define a sequence $s_0,s_1,\ldots,s_d$
  of lengths with $s_d =s$, and
  then prove~$\phi(z)$ by $\Pi^b_1$-$\pind$, where $\phi(z)$ is
  the following~$\Pi^b_1$-formula:
  \begin{equation*}
  \begin{array}{lllll}
  \phi(z) & := & \forall k{\leq}z\ (k{=}|z|{\leq}d \rightarrow \forall A{\in}\UniPoly(k,s_k)\ \\
  & & \;\;\;\;\; (\exists i{\leq}k\ \coef(A,i){\not=}0)
  \rightarrow (\forall u{\in}Z_{A,q}\ \exists i{\leq}k\ h_{A,q}(i){=}u)).
  \end{array}
  \end{equation*}
  Observe that the conclusion of the lemma is the specialization
  of~$\phi(2^d-1)$ to~$k:=d$ and~$A:=P$. 
  The formula $\phi(z)$ uses $d,q$ and the (code of the) sequence $s_0,\ldots,s_d$ 
  as parameters. Note also that $\phi(z)$ is indeed a~$\Pi^b_1$-formula
  because $d,q \in \Log$ and, therefore, 
  all quantifiers except the first two
  are sharply bounded: for $\exists i{\leq}k$, note the condition $k{=}|z|$;
  for $\forall v{\in}Z_{A,q}$ recall that~$Z_{A,q} \subseteq S_q$,
  and~$S_q$ is the set of codes of integers $\{0,1,\ldots,q-1\}\subseteq\formalZ$,
  and these are bounded by $2^{|q|+1} \leq 4q$ (using sign-magnitude representation).
  
  We still need to define $s_k$ for $k = 0,\ldots,d$. We define $s_k$ 
  by the following inverse recurrence relation for $k = d,d-1,\ldots,2,1$:
  \begin{equation}
  \begin{array}{lll}
  s_d & := & s, \\
  s_{k-1} & := & s_{k} + k|q| + |k|.
  \end{array}
  \label{eqn:recursivecaser}
  \end{equation}
  Since $\Log$ is provably closed under addition and multiplication and $d$ is a length, 
  by~$\Pi^b_1$-$\pind$ each~$s_k$ is a length. 
  Also by $\Pi^b_1$-$\pind$, each $s_k$ is bounded by~$s+d^2 |q|+d^2$. 
  Since in the rest of the proof the
  parameter~$q$ will be fixed, we write~$Z_A$ and~$S$ instead of the
  more accurate~$Z_{A,q}$ and~$S_q$. Similarly, we
  write~$h_A$ instead of~$h_{A,q}$.
\bigskip

  To prove $\phi(2^d-1)$ we proceed by $\Pi^b_1$-$\pind$. Concretely, we
  prove $\forall k{<}d (\phi'(k) \rightarrow \phi'(k+1))$ and~$\phi'(0)$,
  where $\phi'(k)$, with $k$ a length, 
  is shorthand notation for~$\phi(2^k-1)$,

\bigskip  

\noindent\textit{Base case:~$\phi'(0)$}.  In this case we
have~$A(\formvar) = a_0$ for a single integer~$a_0$ and, therefore, 
either~$a_0 = 0$, or
else~$Z_{A} = \emptyset$ because~$a_0 \not= 0$ and hence~$h_{A}$ is
vacuously surjective onto~$Z_{A}$.
\bigskip

\noindent\textit{Inductive step:~$\forall k{<}d\;
  (\phi'(k)\rightarrow\phi'(k+1))$}. For convenience of notation, 
we shift the induction parameter~$k$ by one unit, which is clearly
equivalent: assuming~$1 \leq k \leq d$ and~$\phi'(k-1)$, we
prove~$\phi'(k)$.  Fix~$A$ in~$\UniPoly(k,s_k)$,
say~$A(\formvar) = a_0 + a_1 \formvar + \cdots + a_k \formvar^k$.
Assume also that not all coefficients $a_0,a_1,\ldots,a_k$ are zero.
When~$Z_A = \emptyset$ there is nothing to prove as then~$h_{A}$ is
vacuously surjective onto~$Z_A$. Assume
then that~$Z_A \not= \emptyset$ and let then~$v \in S$ be such
that~$A(v) = 0$.  We define the
coefficients~$b_0,b_1,\ldots,b_{k-1}$ of a degree-$(k-1)$
polynomial~$B(\formvar) = b_0 + b_1 \formvar + \cdots + b_{k-1}
\formvar^{k-1}$ and show that
\begin{equation}
A(\formvar) \equiv (\formvar - v)B(\formvar). \label{eqn:goalinFTA}
\end{equation}
This will imply that~$Z_A = Z_B \cup \{v\}$, so the list of roots
of~$A$ will be easily obtained from the list of roots of~$B$. 

The coefficients~$b_0,\ldots,b_{k-1}$ of~$B(\formvar)$ are defined by
the inverse recurrence relation:
\begin{equation}
\begin{array}{lll}
b_k & := & 0, \\
b_{k-i} & := & a_{k-i+1} + b_{k-i+1} v,
\end{array}
\label{eqn:choiceofbi}
\end{equation}
for $i=1,\ldots,k$.

\begin{claim} \label{claim:inunipoly}
The polynomial $B(\formvar)$ belongs to $\UniPoly(k-1,s_{k-1})$.
\end{claim}

\begin{proof}[Proof of Claim~\ref{claim:inunipoly}]
Clearly, there exists a $\PV$-function that, given the appropriate inputs, 
computes the sum of the sequence~$a_{k-i+j} v^{j-1}$ for $j=1,\ldots,i$. 
Using this, by quantifier free induction on~$i = 0,\ldots,k$ we have
\begin{equation}
b_{k-i} = a_{k-i+1} + a_{k-i+2} v + a_{k-i+3} v^2 + \cdots + a_k v^{i-1}. 
\label{eqn:expansionbk}
\end{equation} 
For all~$j=1,\ldots,i$, the~$j$-th term in~\eqref{eqn:expansionbk} has
bit-complexity bounded by~$s_k + k|q|$: indeed~$a_{k-i+j}$ has
bit-complexity at most~$s_k$, and~$v^{j-1} < q^k$.  Thus,
each~$b_{k-i}$ has bit complexity at most~$s_k+k|q|+|k|$, which
equals~$s_{k-1}$ by~\eqref{eqn:recursivecaser}. Thus,
the polynomial $B(\formvar)$ given by the coefficients $b_0,b_1,\ldots,b_{k-1}$
is indeed an element of~$\UniPoly(k-1,s_{k-1})$.
\end{proof}

\begin{claim} \label{claim:notzeropoly}
Not all coefficients of the polynomial $B(\formvar)$ are $0$.
\end{claim}

\begin{proof}[Proof of Claim~\ref{claim:notzeropoly}]
By assumption, not all the coefficients $a_0,a_1,\ldots,a_k$ of $A(\formvar)$
are zero. Therefore, by quantifier-free maximization, there is a largest
$j \leq k$ such that $a_j \not= 0$, so all coefficients~$a_{j+1},a_{j+2},\ldots,a_k$ are $0$. 
By~\eqref{eqn:choiceofbi} and quantifier-free reverse induction,
this means that all coefficients $b_k,b_{k-1},\ldots,b_{j}$ are also $0$. 
Now we argue that $j \geq 1$. Indeed, if~$j=0$, then all $b$s are $0$ so 
by~\eqref{eqn:expansionbk} with $i=k$ we get $A(v) = a_0 + b_0 v = a_0$. 
But $v$ was such 
that $A(v) = 0$, and $a_0$ is non-zero since $j=0$; a contradiction. 
Thus $j \geq 1$. But then~$b_{j-1} = a_j + b_j v = a_j \not= 0$, since $b_j = 0$.
\end{proof}

\begin{claim} \label{claim:same}
$A(\formvar) \equiv (\formvar - v)B(\formvar)$.
\end{claim}

\begin{proof}[Proof of Claim~\ref{claim:same}]
Fix~$u \in \formalZ$. Our goal is to show
that~$A(u) = (u-v)B(u)$. To see this first set the following notation:
\begin{equation}
\begin{array}{lll}
A_0(\formvar) & := & a_{k-1}\formvar^{k-1} + a_k \formvar^k, \\
A_i(\formvar) & := & a_{k-1-i} \formvar^{k-1-i} +  A_{i-1}(\formvar), \\
\end{array}
\label{eqn:choiceofAi}
\end{equation}
for~$i = 1,\ldots,k-1$, and also
\begin{equation}
\begin{array}{lll}
B_0(\formvar) & := & b_{k-1} \formvar^{k-1}, \\
B_i(\formvar) & := & b_{k-1-i} \formvar^{k-1-i} + B_{i-1}(\formvar),
\end{array}
\label{eqn:choiceofBi}
\end{equation}
for~$i = 1,\ldots,k-1$. Note that~$A_{k-1}(\formvar)$ is just
an alternative notation for~$A(\formvar)$, and~$B_{k-1}(\formvar)$ is
an alternative notation for~$B(\formvar)$.
With these, we show, by quantifier-free induction on~$i=0,\ldots,k-1$
with~$u,v,k$ and (the codes of)~$A(\formvar)$ and~$B(\formvar)$ as
parameters, that
\begin{align} 
A_i(u) = (u-v)B_i(u) + (a_{k-1-i} + b_{k-1-i}v) u^{k-1-i}.
\label{eqn:firstequalityhere}
\end{align}
The base
case $i = 0$ of this equation checks out as
\begin{align}
A_0(u) & = a_{k-1} u^{k-1} + a_k u^k \\
& = (u-v)b_{k-1}u^{k-1} + (a_{k-1} + b_{k-1}v) u^{k-1}, \\
& = (u-v)B_0(u) + (a_{k-1} + b_{k-1}v) u^{k-1}
\end{align} 
where the first equality follows from the choice of $A_0(\formvar)$
in~\eqref{eqn:choiceofAi}, the second equality follows from the choice
of $b_{k-1}$ in~\eqref{eqn:choiceofbi}, and the third equality follows
from the choice of $B_0(\formvar)$ in~\eqref{eqn:choiceofBi}.
Similarly, assuming~\eqref{eqn:firstequalityhere} for~$0\leq i\leq
k-2$ as induction hypothesis, the same equation for~$i+1$ checks out
as
\begin{align}
A_{i+1}(u) & = a_{k-2-i} u^{k-2-i} +  A_i(u) \\ 
& =  a_{k-2-i} u^{k-2-i} + (u-v)B_i(u) + (a_{k-1-i} + b_{k-1-i}v) u^{k-1-i} \\
& = (u-v)(b_{k-2-i} u^{k-2-i} + B_i(u)) + (a_{k-2-i} + b_{k-2-i}v) u^{k-2-i}, \\
& = (u-v)B_{i+1}(u) + (a_{k-2-i} + b_{k-2-i}v) u^{k-2-i},
\end{align}
where the first equality follows from the choice of
$A_{i+1}(\formvar)$ in~\eqref{eqn:choiceofAi}, the second equality
follows from the induction hypothesis, the third equality follows from
the choice of $b_{k-2-i}$ in~\eqref{eqn:choiceofbi}, and the fourth
equality follows from the choice of $B_{i+1}(\formvar)$
in~\eqref{eqn:choiceofBi}.  This gives
\begin{align}
A(u) & = A_{k-1}(u) \\
& = (u-v)B_{k-1}(u) + (a_0 + b_0v) \\
& = (u-v) B(u) + A(v) \\
& = (u-v)B(u)
\end{align}
where the first equality follows from the choice of
$A_{k-1}(\formvar)$, the second follows
from~\eqref{eqn:firstequalityhere} with~$i = k-1$, the third follows
from the choices of~$B_{k-1}(\formvar)$ and~$b_0$, and the fourth
follows from the fact that~$v$ is a root of~$A(\formvar)$.  Since~$u$
was an arbitrary value in~$\formalZ$, this proves the claim.
\end{proof}

We are ready to complete the proof of the lemma.  
We argued via Claim~\ref{claim:inunipoly}
that~$B$ is an element of~$\UniPoly(k-1,s_{k-1})$, and
via Claim~\ref{claim:notzeropoly} that not all coefficients
of $B$ are zero. Therefore, the
induction hypothesis~$\phi'(k-1)$ applied to~$B$ gives that the
function~$h_{B} : [k-1] \to S \cup \{q\}$ is surjective
onto~$Z_B$. Think of the function~$h_{B}$ as a
list~$(r_1,\ldots,r_t,q,\ldots,q)$ of~$t \leq k-1$
roots~$r_1 < \cdots < r_t$ of~$B$ in~$S$, followed by a string of
end-of-list markers~$q$. 
We argued via Claim~\ref{claim:same} 
that~$Z_A = Z_B \cup \{v\}$; recall for this that $v$ was from $S$.
Therefore, the
list~$h_{A}$ of roots in~$S$ for~$A$ is obtained from~$h_{B}$ by
inserting~$v$ in the right place of the list of at most~$k-1$ roots
of~$B$, or ignoring it if it is already there, in such a way that the
resulting list is the sorted list of distinct roots of~$A$ in~$S$
followed by end-of-list markers.  It follows that~$h_{A}$ is
surjective onto~$Z_A$, as was to be shown.
\end{proof}

\subsection{Coefficients of Univariate Polynomials}

Next we show that the coefficients of the univariate
polynomial computed by an algebraic circuit with a single indeterminate and
with integer coefficients can be computed in time polynomial in the
representation-size of the circuit, provably in~$\Sonetwo$.

Consider the $\PV$-function $\polynomial(C,d)$ given by
the following polynomial-time algorithm: The input is
a positive integer $d$ written in unary notation, 
and an algebraic circuit~$C$ with a single indeterminate $\formvar$
and syntactic degree bounded by $d$. 
The circuit can be thought of as given by a straight-line program of 
gates~$C_0,C_1,\ldots,C_t$, identified with their labels. 
For all gate numbers~$u = 0,1,\ldots,t$ in this order, 
the algorithm 
constructs the coefficients $a_{u,0},\ldots,a_{u,d}$
of a degree-$d$ polynomial $P_u \in \UniPoly(d)$, inductively
according to the following rules:
\begin{enumerate} \itemsep=0pt
\item if $C_u = c \in \formalZ$,
then $a_{u,0} := c$ and $a_{u,k} := 0$ for $k = 1,2,\ldots,d$,
\item if $C_u = \formvar$, then $a_{u,0} := 0$
and $a_{u,1} := 1$, and $a_{u,k} := 0$ for $k = 2,3,\ldots,d$,
\item if $C_u = C_v + C_w$ with $v<u$ and $w<v$,
then $a_{u,k} := a_{v,k} + a_{w,k}$ for $k = 0,1,\ldots,d$,
\item if $C_u = C_v \times C_w$ with $v < u$ and $w<u$,
then $a_{u,k} := \sum_{t=0}^k a_{v,t} a_{w,k-t}$ for $k = 0,1,\ldots,d$.
\end{enumerate}
On termination, the algorithm outputs the code of the polynomial $P_t$.

We verify that the algorithm runs in polynomial time.
Let $s$ be the bit-complexity of the constants plugged into the
circuit~$C$. By the definition of the representation-size 
of algebraic circuits,~$s$ 
is bounded by the representation-size of~$C$.
The bit-complexity $s_u$ of the coefficients~$a_{u,0},\ldots,a_{u,d}$,
and the syntactic degree $d_u$ of the circuit rooted at $u$,
obey the following recurrences:
\begin{enumerate} \itemsep=0pt
\item $s_u \leq s$ and $d_u = 1$ if $C_u = c \in \formalZ$,
\item $s_u \leq 2$ and $d_u = 1$ if $C_u = \formvar$,
\item $s_u \leq 1+\max\{s_v,s_w\}$ and $d_u \leq \max\{d_v,d_w\}$ 
if $C_u = C_v + C_w$,
\item $s_u \leq |d+1|+(s_v+s_w)$ and $d_u \leq d_v + d_w$ if 
$C_u = C_v \times C_w$.
\end{enumerate}
Therefore, by induction on $u=0,1,\ldots,t$, we get
$s_u \leq (s+|d+1|)(2 d_u - 1)$.
Since $d_u \leq d$, the conclusion is that each $s_u$ is bounded by a 
polynomial $\poly(d,s)$, so 
the algorithm runs in time polynomial in the size of its input.

\begin{lemma}[Univariate Coefficient-Extraction Lemma in~$\Sonetwo$] \label{lem:coefficientslemma}
For every $d \in \Log$, every algebraic circuit $C$ with
a single indeterminate input,  
syntactic degree at most~$d$, and integer constants,
we have~$P \equiv C$, for~$P = \polynomial(C,d)$.
\end{lemma}

\begin{proof}
    Let $C_0,C_1,\ldots,C_t$ be the straight-line program of $C$
    and let $\overline{p}$ be the integer constants plugged into $C$.
    Let $P_0,P_1,\ldots,P_t$ be the degree-$d$ polynomials constructed
    by the algorithm that defines~$\polynomial(C,d)$. 
    For every $z \in \formalZ$, fixed as parameter of induction
    together with $\overline{p},d$ and the sequences $C_0,C_1,\ldots,C_t$ and $P_0,P_1,\ldots,P_t$,
    prove $P_u(z) = C_u(z)$, i.e.,
    \begin{equation}
        \evalarithmetic(P_u,z,\overline{p},d) = \evalarithmetic(C_u,z,\overline{p},d),
        \label{eqn:induz}
    \end{equation}
    by quantifier-free induction on $u=0,1,\ldots,t$. For this, use
    the rules 1--4 that define the algorithm of the $\PV$-function~$\polynomial(C,d)$.
    Note that this works for any value $z \in \formalZ$, regardless of
    its bit-complexity. Thus, applying~\eqref{eqn:induz} to $t = u$ we get 
    $P(z) = P_t(z) = C_t(z) = C(z)$. Since $z$ was arbitrary in $\formalZ$,
    this shows $P \equiv C$ and the lemma is proved.
\end{proof}

\subsection{The Coding Argument and its Formalisation}


We show that the theory~$\Sonetwo$ proves the \SZ\ Lemma
for multivariate polynomials of polynomial degree. As in the univariate case, 
in the $n$-dimensional case we consider the roots with all its entries
in $S_q$ for given $q$, where $S_q = \{0,1,\ldots,q-1\}$. For any 
non-zero polynomial~$P(\overline\formvar)$ with~$n$
variables and individual degree at most~$d$, any given
non-root~$\va$ of~$P(\overline\formvar)$, and any 
given positive integer~$q$, we build a function that has
domain-size~$n \cdot d \cdot q^{n-1}$ and covers, in its range, all the roots
of~$P(\overline\formvar)$ within~$S_q^n$. 
Note that in this case the covering function
must be given access to one non-root $\va$ of~$P(\overline\formvar)$ that
certifies it as non-zero. In other words, the statement that the
theory~$\Sonetwo$ will formalise is the following:

\begin{quote}
  \emph{if~$P(\overline\formvar)$ has $n$ variables, individual degree $d$, and at least one non-root $\va$ in $\formalZ^n$, \\
  then $P(\overline\formvar)$ has at least~$q^n \cdot (1-nd/q)$ non-roots
  in~$S_q^n$}. 
\end{quote}
To formulate and prove this statement with the required precision, we
need to introduce some notation.

Let $n,d,q$ be small positive integers. Let~$P$ be an element
of~$\Circuits(n,d)$, that is,~$P$ is (the code of) an algebraic
circuit with~$n$ inputs and syntactic individual degree at most~$d$. As
in the univariate case, define
\begin{equation*}
\begin{array}{lcl}
S_q & := & \{0,\ldots,q-1\} \subseteq \Z,  \\
Z_{P,q} & := & \{ (u_1,\ldots,u_n) \in S_q^n : P(u_1,\ldots,u_n) = 0 \}, \\
C_{n,d,q} & := & [n] \times [d] \times S_q^{n-1}.
\end{array}
\end{equation*}
Note that~$Z_{P,q}$ is the set of
roots, or \emph{zeros}, of~$P$ in the set~$S_q^n$.  We will use $C_{n,d,q}$ as the set
of \emph{codes} of roots of $P$ in the set $S_q^n$. Observe that if $q > nd$, then $|C_{P,d,q}| < |Z_{P,q}|$
and compression is achieved. 

For the fixed $n,d,q,P$ and
any fixed additional vector~$\av = (a_1,\ldots,a_n)$ of integers (not necessarily in $S_q^n)$
we define a pair of~$\PV$-functions
\begin{align*}
& \encode_{P,\av} : S_q^n \to C_{n,d,q}, \\
& \decode_{P,\av} : C_{n,d,q} \to S_q^n.
\end{align*}
These functions are defined by the polynomial-time algorithms that compute them. 

\paragraph{Encoding.} The input to~$\encode_{P,\av}(\bv)$ 
is a $5$-tuple $(P,d,q,\av; \bv)$ with $P,d,q,\av$ as above, with $d$ and $q$ presented
in unary notation, and $\bv = (b_1,\ldots,b_n)$ a vector $S_q^n$.
The goal is to encode the vector~$\vb$ as an element in $C_{n,d,q}$. 
The algorithm first evaluates $P$ at $\av$ and $\bv$ to check 
if~$P(\av) \not= 0$ and~$P(\bv) = 0$. 
If this fails, then either the given $\av$ is not suitable for the encoding scheme, 
or the given~$\bv$ is not even a root, which we do not need to encode. In such a case the
algorithm returns the default value~$(1,1,\zerov)$ from $[n] \times [d] \times S_q^{n-1}$, 
where $\zerov = (0,\ldots,0)$ 
is the all-zero vector in~$S_q^{n-1}$. If indeed~$P(\av) \not = 0$ 
and~$P(\bv) = 0$, then the algorithm loops through 
$k = 1,2,\ldots,n$ in this order until it finds the first position where
\begin{equation}
\begin{array}{lllcl}
P(a_1,\ldots,a_n) &&& \not= & 0 \\ 
P(b_1,a_2,\ldots,a_n) &&& \not= & 0 \\
P(b_1,b_2,a_3,\ldots,a_n) &&& \not= & 0 \\ 
\;~~~~~~~~~~\vdots & \ddots && \vdots & \vdots\\ 
P(b_1,b_2,b_3,\ldots,b_{k-1},& a_{k} &,a_{k+1},\ldots,a_n) & \not= & 0 \\  
P(b_1,b_2,b_3,\ldots,b_{k-1},& b_{k} &,a_{k+1},\ldots,a_n) & = & 0.
\end{array}
\label{eqn:hybridthis}
\end{equation}
A first such $k \in [n]$ must exist since $P(\av) \not= 0$ and $P(\bv) = 0$. 
Next, consider the \emph{univariate} algebraic circuit
\begin{equation}
Q(\formvar) :=
P(b_1,\ldots,b_{k-1},\;\;\formvar,\;\;a_{k+1},\ldots,a_n) \label{eqn:Qpol}
\end{equation} 
By definition $Q(b_k) = 0$ so $b_k$ is a root of $Q(\formvar)$.
Then the algorithm loops 
through $r \in S_q$ to form the ordered list of roots of $Q(\formvar)$ in $S_q$
and finds the position $i \in [q]$ of $b_k$ in this ordered list. If this 
index exceeds $d$ (spoiler: this will never happen but the algorithm need not know this),
then something went wrong and the algorithm outputs the default value $(1,1,\zerov)$
as before. Otherwise, it outputs~$(k,i,b_1,\ldots,b_{k-1},b_{k+1},\ldots,b_n)$, 
which is in $[n] \times [d] \times S_q^{n-1}$.

\paragraph{Decoding.} The input to $\decode_{P,\av}(k,i,\cv)$ is
the~$7$-tuple~$(P,d,q,\av; k,i,\cv)$, with~$P,d,q,\av$ as above, with~$d$ and~$q$ presented in unary 
notation,~$k \in [n]$ and $i \in [d]$, and~$\cv = (c_1,\ldots,c_{n-1})$ a vector 
in~$S_q^{n-1}$. The algorithm first evaluates~$P$ on the vector~$\av$
to check that~$P(\av) \not= 0$. If this check fails,
then the given~$\av$ is not suitable for the encoding scheme. In such a case the 
algorithm outputs the default value $\zerov$ from $S_q^n$, where $\zerov = (0,\ldots,0)$ is
the all-zero vector of length $n$. If indeed~$P(\av) \not= 0$, 
consider the following univariate algebraic circuit:
\begin{equation}
G(\formvar) :=
P(c_1,\ldots,c_{k-1},\;\; \formvar\;\;,a_{k+1},\ldots,a_n). \label{eqn:Gpol}
\end{equation} 
The algorithm loops through $r$ in $S_q$ to check if there are 
at least $i$ different $r \in S_q$ such that~$P(r) = 0$.
If it finds less than $i$ roots in $S_q$, then something went wrong and again the 
algorithm outputs~$\zerov = (0,\ldots,0)$
in~$S_q^n$ as a default value. Otherwise, it sets~$r$ to the~$i$-th smallest
root of~$G(\formvar)$ in~$S$ and outputs~$(c_1,\ldots,c_{k-1},r,c_{k},\ldots,c_{n-1})$ in $S_q^n$.

\paragraph{Correctness of the encoding scheme.}
We are ready to prove that the encoding scheme for roots of $P$ given
by the pair of functions $\encode_{P,\av}$ and $\decode_{P,\av}$ is correct. 
This correctness hinges on the condition $P(\av) \not= 0$
but does not require any upper bound on the components of the vector $\av$.
Note also that the 
functions $\encode_{P,\av}$ and $\decode_{P,\av}$ are polynomial-time
computable because they are given $P,d,q,\av$ in the input, with $d$ and $q$ 
in unary notation. Finally, it is also worth pointing out that 
the correctness of the encoding scheme will not require 
any lower bound on $q$. However, as noted earlier, the encoding achieves
some compression only if $nd < q$.

In the statement of the next lemma, which states the correctness, recall that 
we use the notation~$P \equiv 0$ to the
denote the (unbounded) formula $\forall \av{\in}\formalZ^n\ P(\av){=}0$,
where $\formalZ$ denotes the set of encodings of \emph{all} integers.

\begin{theorem}[Proof of \SZ\ Lemma in~$\Sonetwo$]\label{lem:encoding-roots} 
  For all~$n,d,q \in \Log$ and all~$P \in \Circuits(n,d)$, either $P \equiv 0$ or
  for every~$\av = (a_1,\ldots,a_n) \in \formalZ^n$ such that $P(\av) \not= 0$ the
  function~$\decode_{P,\av}$ is surjective onto~$Z_{P,q}$ and inverts
  $\encode_{P,\av}$ on $Z_{P,q}$: for every $\bv \in Z_{P,q}$ we have 
  $\decode_{P,\av}(\encode_{P,\av}(\bv)) = \bv$.
\end{theorem}

\begin{proof}
  Since for this proof the parameters~$n,d,q,P,\av$ will be fixed, 
  we drop them from the
  notation and write~$S$ and~$Z$ instead of~$S_q$
  and~$Z_{P,q}$, and $\decode$ and $\encode$ 
  instead of~$\decode_{P,\av}$ and~$\encode_{P,\av}$. 
  Given a root~$\bv = (b_1,\ldots,b_n) \in Z$, our goal is to find a
  triple~$(k,i,\cv) \in [n] \times [d] \times S^{n-1}$ such
  that~$\decode(k,i,\cv) = \bv$; indeed 
  we will show that the choice~$(k,i,\cv) = \encode(\bv)$ works,
  and this proves also the second part of the lemma 
  that~$\decode(\encode(\bv)) = \bv$. 
  
  Fix~$(k,i,\cv) = \encode(\bv)$. To see that~$\decode(k,i,\cv) = \bv$, recall 
  that in the definition of $\encode(\bv)$ the 
  index~$k$ is defined as the smallest index in the sequence~$1,\ldots,n$ such~\eqref{eqn:hybridthis} holds.
  In particular, the univariate circuit~$Q(\formvar)$ from~\eqref{eqn:Qpol}
does not vanish everywhere as it is non-zero on~$a_{k}$. Further,~$Q(\formvar)$
has syntactic degree at most~$d$ since it is a restriction of $P$
with plugged constants from $\va$.
Using the~$\PV$-function of
Lemma~\ref{lem:coefficientslemma}, we can extract the coefficients of
a polynomial~$H = \polynomial(Q,d)$ in~$\UniPoly(d)$ such 
that~$H \equiv Q$; i.e., the
equality~$H(u) = Q(u)$ holds for all integers~$u \in \formalZ$. 
In particular~$H(a_k) = Q(a_k) \not= 0$ and, by~\Cref{lem:univariate}, 
the function~$h_{H,q}$ in surjective
onto the set of roots of~$H(\formvar)$ in~$S$. Also~$b_{k}$ is one 
of these roots
by~\eqref{eqn:hybridthis} and the just mentioned 
fact that~$H \equiv Q$.

Now recall that in the definition of the
function~$\encode(\bv)$, the components~$i$ and~$\cv$ in
its output are defined
to satisfy the following conditions:
\begin{quotation}
\noindent C1. element $b_{k} \in S$ is
the~$i$-th smallest root of~$Q(\formvar)$ in~$S$ \emph{unless}~$i > d$, \\
\noindent C2. vector~$\cv = (c_1,\ldots,c_{n-1})$ is set to~$(b_1,\ldots,b_{k-1},b_{k+1},\ldots,b_n)$.
\end{quotation}
The~$i$-th smallest root of~$Q(\formvar)$ in $S$ is, by virtue
of~$H \equiv Q$, also the $i$-th smallest root
of~$H(\formvar)$ in $S$. Thus, by Lemma~\ref{lem:univariate} and
the fact that~$H(b_k) = Q(b_k) = 0$, this
is the preimage of~$b_{k}$ under~$h_{H,q}$, which is at most~$d$
since the domain of~$h_{H,q}$ is~$[d]$. This means
that, in~C1, the index~$i$ is indeed such that~$b_{k}$ is the~$i$-th smallest
root of~$Q(\formvar)$ in~$S$. We show that~$\decode(k,i,\cv) = \bv$.

Let~$r$ be the~$i$-th smallest root in~$S$ of 
the univariate polynomial~$G(\formvar)$ from~\eqref{eqn:Gpol},
which exists because by~C2 we have~$G \equiv Q$, even 
syntactically as circuits in this case.
By the discussion in the previous paragraph~$b_{k}$ is 
the~$i$-th smallest root of~$Q(\formvar)$ in $S$. By the definition
of the function~$\decode(k,i,\cv)$ we have~$\decode(k,i,\cv) = 
(c_1,\ldots,c_{k-1},r,c_k,\ldots,c_{n-1})$, which by~C1 and~C2 
equals~$(b_1,\ldots,b_{k-1},b_k,b_{k+1},\ldots,b_n)$, also known as~$\bv$.
This completes the proof.
\end{proof}

\section{Hitting Sets and Identity Testing in the Theory}

\subsection{Hitting Sets}

So far we argued that the theory~$\Sonetwo$ proves that
every~$n$-variable algebraic circuit with integer coefficients 
and small degree has
relatively few roots in~$S_q^n$, unless the
polynomial vanishes everywhere. By a standard counting
argument, it follows that for any given bounds~$d$ and~$m$ 
on the degree and the description size of algebraic circuits
with~$n$ variables, there is a set~$(\cv_1,\ldots,\cv_r)$ of~$r =
\poly(n,d,m)$ points in~$S_q^n$, with~$q = \poly(n,d)$, that intersects the set 
of non-roots of every non-vanishing algebraic circuit with those bounds. We 
say that~$(\cv_1,\ldots,\cv_r)$ is a \emph{hitting set
over $S_q$}. Formally:

\begin{definition}[In $\Sonetwo$] \label{def:hit}
Let $\scriptC$ be a definable class of algebraic circuits (\Cref{def:definable classes of algebraic circuits}).
Let~$e$ be a parameter, let $n,d,s,q,r,m \in \Log$ be positive lengths and 
let~$H = (\vh_1,\ldots,\vh_r) \in (S_q^n)^r$ be an $r$-sequence of $n$-vectors
with entries in $S_q$.
We say that~$H$ is a \emph{hitting set for~$\scriptC_e(n,d,s,m)$
over~$S_q$} if and only if for every~$P \in \scriptC(n,d,s,m)$, if
there exists~$\av \in \formalZ^n$ such that~$P(\av) \not= 0$, then there
exists~$i \in [r]$ such that~$P(\vh_i) \not= 0$.  
\end{definition}

If~$H$ as above is not a hitting set, then a \emph{witness} that $H$ is not a hitting 
set is a pair
\begin{equation}
    (P,\av) \in \scriptC_e(n,d,s,m) \times \formalZ^n \label{eqn:witnesspair}
\end{equation}
such that~$P(\av)
\not= 0$ and $P(\vh_j)=0$ for every~$j \in [r]$.  The \emph{size} of the
hitting set~$H$ is the number of
distinct tuples among the~$\vh_i$, which is always at most~$r$.
In what follows we use the results of the previous sections to show
that the counting argument that proves that polynomial-size hitting
sets exist can be formalised in the theory~$\Sonetwo+\dWPHP(\PV)$.

\paragraph{Small witnesses.} We start by showing that, if~$q$ is
sufficiently big, but polynomial in $n$ and $d$, then
the witnesses of failure as in~\eqref{eqn:witnesspair} can always be chosen 
with $\va$ in $S^n_q$.
This is interesting in its own right:

\begin{lemma}[Small Witnesses Exist in $\Sonetwo+\dWPHP(\PV)$] \label{lem:smallwit}
For all $n,d,q \in \Log$ and every algebraic circuit~$P \in \Circuits(n,d)$,
if $P \not\equiv 0$ and $q \geq 2dn$, then there exists $\va \in S^n_q$
such that~$P(\va) \not= 0$.
\end{lemma}

\begin{proof}
Assume $P \not\equiv 0$, so there exists $\va \in \formalZ^n$ such that $P(\va) \not= 0$.
We find an~$\va_0 \in S_q^n$ also with~$P(\va_0) \not= 0$.
By Theorem~\ref{lem:encoding-roots}, the function~$\decode_{P,\va} : C_{n,d,q} \rightarrow S_q^n$
is surjective onto~$Z_{P,q}$. Let~$b = q^n$ and note that the parameters~$q,n$ 
allow us to code each point~$\va_0 \in S_q^n$ as an~$n$-tuple of elements 
in~$\{0,\ldots,q-1\}$. Thus, the co-domain~$S_q^n$ of~$\decode_{P,\va}$ can be identified
with the set~$[q^n]$. Similarly, the parameters~$n,d,q$ allow us to identify  
the domain~$C_{n,d,q}$ with the set~$[dnq^{n-1}]$. Thus, 
the function~$\decode_{P,\va}$ can
be thought of as the parameterization of a~$\PV$-function~$f$ such that:
\begin{equation}
    f_{P,\va,n,d,q} : [a] \to [b] \;\;\;\;\text{ where }\;\;\;\;
    \begin{array}{lllll} a & := & dnq^{n-1}, \\ b & := & q^n \,.\end{array} \label{eqn:specone}
\end{equation}
Now note that the assumption $q \geq 2dn$ implies $b \geq 2a \geq 2$. Since each
element of $[b]$ codes an element of $S_q^n$, and each element of $C_{n,d,q}$ is coded
by an element of $[a]$, by $\dWPHP^a_b(f_{P,\va,n,d,q})$, there exists 
$\va_0 \in S_q^n$ that is outside the range~$Z_{P,q}$ of~$\decode_{P,\va}$. Since~$Z_{P,q}$
is the set of all roots of $P$ in $S_q^n$, we get that $\va_0$ is a non-root
of~$P$ that belongs to~$S_q^n$, as needed.
\end{proof}

As stated earlier, 
one consequence of Lemma~\ref{lem:smallwit} is that whenever $H \subseteq S_q^n$ fails 
to be a hitting set in the sense of~Definition~\ref{def:hit}, if $q \geq 2dn$
then not only is there a witness of this failure as in~\eqref{eqn:witnesspair}, but 
there is even a witness of the following form:
\begin{equation}
    (P,\va) \in \mathscr{C}(n,d,s,m) \times S_q^n \label{eqn:restrictedwit}
\end{equation}
Next we show how to use this to get hitting sets with one more application of $\dWPHP$.

\paragraph{Hitting sets exist.}
Let $\scriptC$ be a definable class of algebraic circuits,
let $e$ be a parameter, and 
let~$n,s,d,q,r,m$ be positive integers with~$q > d$ and~$s \geq n$.
Consider the function
\begin{equation}
g_{\scriptC,e,n,d,s,q,r,m} : 
\scriptC_e(n,d,s,m) \times S_q^n \times ([n] \times [d] \times S_q^{n-1})^r \to (S_q^n)^r 
\label{eqn:specification}
\end{equation}
defined by
\begin{equation}
(P,\av,\cv_1,\ldots,\cv_r)
\mapsto
(\decode_{P,\av}(\cv_1),
\ldots,\decode_{P,\av}(\cv_r)), \label{eqn:defofg}
\end{equation}
where $\cv_1,\ldots,\cv_r \in [n] \times [d] \times S_q^{n-1}$ are candidate
codes for roots of $P$ in $S_q^n$.
Recall that if $P(\av) = 0$,
then the definition of $\decode_{P,\av}(\cv)$ is such that its output
is $\zerov$, so that $g(P,\av) = (\zerov,\ldots,\zerov)$ in this case.
Here, the all-zero tuple~$(\zerov,\ldots,\zerov)$ is
used as a \emph{don't
care} value indicating that~$\av$ is not serving the purpose of 
certifying that~$P$ does not vanish everywhere. 
%

\begin{lemma}[In~$\Sonetwo + \dWPHP(\PV)$] \label{lem:outsiderange}
For every definable class $\scriptC$ of algebraic circuits, every parameter $e$
for its decoding function,
all~$n,d,s,q,r,m \in \Log$, and every~$H \in (S_q^n)^r$, if~$q \geq 2dn$ 
and~$H \not\in \Img(g)$, where~$g =
g_{\scriptC,e,n,d,s,q,r,m}$, then~$H$ is a hitting set for~$\scriptC_e(n,d,s,m)$
over~$S_q$.
\end{lemma}

\begin{proof}
We prove the contrapositive statement. Assume that~$H
= (\dv_1,\ldots,\dv_r)$ is not a hitting set and let~$(P,\av)$ witness
it as in~\eqref{eqn:witnesspair}: $\av \in \mathbb{Z}^n \setminus
Z_{P,q}$ and $\dv_j \in Z_{P,q}$ for every $j \in [r]$.
Since~$q \geq 2dn$, by~Lemma~\ref{lem:smallwit} we may assume
that $\va$ is in $S_q^n$; i.e., the witness $(P,\va)$ is as
in~\eqref{eqn:restrictedwit}: $\va \in S_q^n \setminus Z_{P,q}$ and
$\vd_j \in Z_{P,q}$ for every $j \in [r]$.
By Lemma~\ref{lem:encoding-roots}, there 
exists~$(\cv_1,\ldots,\cv_r) \in ([n] \times [d] \times
S_q^{n-1})^r$ such that~$\decode_{P,\av}(\cv_j) = \dv_j$ for every~$j
\in [r]$.  Hence,~$g(P,\av,\cv_1,\ldots,\cv_r) =
(\dv_1,\ldots,\dv_r)$ by the definition of~$g$
in~\eqref{eqn:defofg}. This shows that~$H$ is in the range of~$g$.
\end{proof}

We intend to apply $\dWPHP(g)$ for the function~$g$ in~\eqref{eqn:specification}.
To do so we need to view~$g$ as a function
from~$[a]$ to~$[b]$ for appropriate $a$ and $b$ such that $b \geq 2a \geq 2$.
In other words, we need to specify the numeric encodings for the
elements in the domain and the co-domain of~$g$. 
For the co-domain,
the parameters~$n,q,r$ allow us to code each
element $H \in (S_q^n)^r$ as an~$nr$-tuple of elements
in~$\{0,\ldots,q-1\}$, so the co-domain can be identified with the
set~$[q^{nr}]$. Recall that the parameters~$n,r$ are lengths,
i.e., elements of~$\Log$, so~$q^{nr}$ exists. For the
domain, the parameters $e,n,d,s,q,r,m$ allow us to code each
element~$(P,\av,\cv_1,\ldots,\cv_r) \in \scriptC_e(n,d,s,m)
\times S_q^n \times ([n] \times [d] \times S_q^{n-1})^r$ as a number in~$[2^m
  n^r d^r q^{(n-1)r+n}]$. Again this quantity exists because the
parameters $n,r,m$ are lengths; i.e., elements of~$\Log$. Summarizing, the
specification~\eqref{eqn:specification} becomes
  \begin{align}
  g_{\scriptC,e,n,d,s,q,r,m} : [a] \to [b] \;\;\;\;\text{ where } \;\;\;
  \left. { \begin{array}{lll}
  a & := & 2^m n^r d^r q^{(n-1)r+n}. \\
  b & := & q^{nr}. \end{array} }\right. \label{eqn:uandv}
  \end{align} 
The fact that every element in $[b]$ denotes an element of the co-domain 
will be used in the proof of the next theorem. 

\begin{theorem}[Small Hitting Sets Exist in~$\Sonetwo + \dWPHP(\PV)$] \label{lem:hittingset} For
  every definable class $\scriptC$ of algebraic circuits, every parameter $e$ for
  its decoding function,
  and all~$n,d,s,q,r,m \in \Log$ such that~$q \geq 2dn$ and~$r > m + n|q|$, there
  exists~$H \in (S_q^n)^r$ such that~$H$ is a hitting set
  for~$\scriptC_e(n,d,s,m)$ over~$S_q$.
\end{theorem}

\begin{proof}
  To show the existence of~$H = (\dv_1,\ldots,\dv_r)$
  as claimed in the lemma we apply the dual weak pigeonhole
  principle~$\dWPHP(g)$ on the~$\PV$-function~$g$ in~\eqref{eqn:specification},
  with the parameters indicated there.
  By the discussion preceding the lemma we have~$g :
  [a] \to [b]$ where~$a$ and~$b$ are as in~\eqref{eqn:uandv}.
  Comparing~$a$ to~$b$, we see that
  the assumptions~$q \geq 2nd$ and~$r > m + n|q|$ 
  yield~$b \geq 2a$. Indeed,~$m,n,|q|$ are integers, so~$r \geq m+1+n|q|$.
  Therefore, using $q \geq 2nd$, we get $(q/(nd))^r \geq 2^{m+1+n|q|} 
  \geq 2^{m+1} q^n$. Multiplying both sides by $q^{nr}$ and rearranging we get $b \geq 2a$. 
  Recalling that each element of~$[b]$ codes
  an element of~$(S_q^n)^r$, by~$\dWPHP^a_b(g)$ 
  there exists~$H \in (S_q^n)^r$ that is outside the range
  of~$g$.  By Lemma~\ref{lem:outsiderange},
  this~$H$ is a hitting set for~$\scriptC_e(n,d,s,m)$ over~$S_q$, as
  needed.
\end{proof}

\subsection{$\PIT$ in $\coRP$ and in $\Ppoly$}

Let \emph{PIT} stand for \emph{Polynomial
Identity Testing}, which is the following computational problem:
\begin{quote}
  \PIT: \emph{Given positive integers $n,d,s$ written in unary notation
  and given an algebraic circuit~$C$ of representation size~$s$,
  with~$n$ variables, integer constants, and syntactic individual degree~$d$,
  does it compute the identically zero polynomial?}
\end{quote}
The requirement to have~$d$ written in unary notation in the input
is a way to enforce that the polynomial computed by the circuit
has \emph{polynomial degree}; the size of the input is inflated at
the same rate as the degree. Similarly, since our definition of
degree includes the parametric inputs, the size of the input
is inflated at the same rate as the bit-complexity of the constants
it computes.

In the statement of the problem, the phrase ``identically zero polynomial'' can be
interpreted in two ways: syntactically and semantically. In the
syntactic interpretation, it is the polynomial which has zero coefficient
on each monomial when written explicitly in its unique representation as a finite
linear combination of monomials. In the semantic interpretation, it is the polynomial
that evaluates to zero on all points in $\mathbb{Z}^n$. It is well known that
both interpretations yield the same computational problem, but note that
weak theories such as $\Sonetwo$ or even $\Sonetwo + \dWPHP(\PV)$ may
not be able to prove this since the standard proof involves exponential-time
computations.

The Schwartz-Zippel Lemma 
as stated in Theorem~\ref{thm:SZ-meta}
implies that the semantic interpretation of~\PIT\ is not only decidable
but even in~$\coNP$: it
suffices to check that~$C$ evaluates to zero on all 
points~$\av = (a_1,\ldots,a_n)$ of the cube~$\{0,\ldots,q-1\}^n$,
where~$q = 2nd$. For the sake of
formalisation in weak theories, 
let us state what we call the \emph{explicitly~\coNP\ formulation} of the problem:
\begin{quote}
  \PITCONP: \emph{Given positive
  integers~$n,d,s$ written in unary notation and given an algebraic
  circuit~$C$ of representation size~$s$, with~$n$ variables, integer constants, 
  and syntactic individual degree~$d$, does it 
  evaluate to~$0$ on all
  points~$(a_1,\ldots,a_n)$ of the cube~$\{0,\ldots,q-1\}^n$ of side-length~$q = 2nd$?}
\end{quote}
This version of the problem is in~$\coNP$ yet the
Schwartz-Zippel Lemma shows that it is even in~$\coRP$ and hence
in~$\Ppoly$. Indeed, if~$P(\av) \not= 0$ for at least one point~$\av \in
\{0,\ldots,q-1\}^n$, then~$P(\av) \not= 0$ for a fraction of at
least~$1-nd/q \geq 1/2$ of the points~$\av \in \{0,\ldots,q-1\}^n$. 
This means that an evaluation at a random
point in~$\{0,\ldots,q-1\}^n$ has chance at least $1/2$
of detecting that~$C$ does not always evaluate to zero. This
is~$\coRP$. The results in this paper show that these
statements are available in the theory~$\Sonetwo + \dWPHP(\PV)$.
For $\Ppoly$ the definition of what this means is straightforward.
Let Boolean circuits be encoded in the theory in
a completely analogous way to how algebraic circuits
are encoded in the theory. The only difference is that
the evaluation function for Boolean circuits
takes only binary strings as inputs for its variables or parameters,
and interprets its gates in the Boolean sense, with $+$ 
as disjunction, and $\times$ as conjunction.

\begin{definition}[In~$\Sonetwo$]
Let $c \geq 1$ be a standard constant
and let $A$ be a definable set of strings 
given by a formula $\varphi(x)$ without
parameters, i.e., with all free variables indicated. The
set $A$ is in~$\SIZE(n^c)$ if for every $n \in \Log$ there
exists a Boolean circuit $C$ with $n^c$ gates, with $n$ inputs and $1$ output, 
such that for every~$x \in \{0,1\}^n$ we have that~$C(x)=1$ if and only
if~$\varphi(x)$ holds.
\end{definition}

Membership of $\PIT$ in $\Ppoly$ is now
an immediate consequence of the existence of
hitting sets. We state this explicitly.
Let $\PIT$ denote the $\Pi_1$-definable set of
strings that encode the~YES instances of the computational
problem defined earlier: the strings
encode~$4$-tuples~$(C,n,d,s)$
where~$n,d,s$ are integers presented in unary notation
(i.e., as three strings of the form $100\cdots 0$ with 
lengths $n,d,s$),
and $C$ is an algebraic circuit in $\Circuits(n,d,s)$
such that $C \equiv 0$ holds; i.e., 
the~$\Pi_1$-formula~$\forall \va{\in}\formalZ^n\ C(\va){=}0$ holds. 
Similarly, let~$\PITCONP$ denote the~$\Pi^b_1$-definable
set of strings that encode the~YES instances of the
explicitly~$\coNP$ version of~$\PIT$.
We say that a theory proves that a definable set~$A$ 
of strings is in~$\Ppoly$
if there exists a standard constant~$c > 0$ such
that it proves that~$A$ is in~$\SIZE(n^c)$.

\begin{theorem}[$\PIT$ is in $\Ppoly$ in $\Sonetwo+\dWPHP(\PV)$]
The theory $\Sonetwo+\dWPHP(\PV)$ proves that both~$\PIT$ 
and $\PITCONP$ are in~$\Ppoly$. Indeed, there is a standard
constant $c \geq 1$, such that it proves the 
following stronger statement proving both claims simultaneously:
For all~$n,s,d,q \in \Log$ such that~$q \geq 2nd$, there
exists a Boolean circuit~$C$ with~$(n+d+s+q)^c$ gates such that for
every~$P \in \Circuits(n,d,s)$ the following hold:
\begin{enumerate} \itemsep=0pt
\item if $C(P)=1$, then for every $\av \in \formalZ^n$ we have $P(\av) = 0$,
\item if $C(P)=0$, then there exists $\av \in S_q^n$ such that $P(\av) \not= 0$.
\end{enumerate}
\end{theorem}

\begin{proof}
  Consider the $\PV$ function that
  given $(C,n,d,s)$ as required in the input for $\PIT$
  and given~$H = (\vh_1,\ldots,\vh_r) \in S_q^n$, evaluates $C$ on all points in $H$
  and outputs $1$ if all evaluations are zero, and outputs $0$ otherwise. 
  Setting~$q=2nd$ and~$r = s+n|q|$, and hardwiring for~$H$ the hitting set 
  for~$\Circuits(n,d,s)$ of Theorem~\ref{lem:hittingset} we get a Boolean 
  circuit $C$ as required, for a standard $c \geq 1$ that depends only
  on the runtime of the $\PV$ function.
\end{proof}

\begin{corollary}[Small Counterexamples for Identities in $\Sonetwo+\dWPHP(\PV)$]
    The theory $\Sonetwo + \dWPHP(\PV)$ proves that $\PIT$ and
    the explicitly~$\coNP$ version of~$\PIT$ are the same problem, i.e.,
    the formulas that define the sets~$\PIT$ and~$\PITCONP$ are equivalent.
\end{corollary}

Stating membership of $\PIT$ in~$\coRP$ within the theory $\Sonetwo+\dWPHP(\PV)$
is much less straightforward, but at least the necessary 
ingredients are now available. 
We refer the reader to Je\v{r}\'abek's seminal 
papers~\cite{Jer04,Jer07} 
where the complexity classes $\BPP$ and $\RP$ are studied 
in~$\PV_1 + \dWPHP(\PV)$, a subtheory of $\Sonetwo+\dWPHP(\PV)$.

\section{Complexity of Finding Hitting Sets} \label{sec:complexity}

\subsection{Hitting Set Axioms}

Let~$\scriptC$ be a definable class of algebraic circuits
with parameterized slices $\scriptC_e(n,d,s,m)$ according
to the definition. To recall, the slice $\scriptC_e(n,d,s,m)$
is the set of algebraic circuits in
the class~$\scriptC$ with parameter $e$ 
that have at most~$n$ indeterminates, syntactic individual
degree at most~$d$, representation size at most $s$,
but description size~$m \leq s$ as members of~$\scriptC$. 
Let the \emph{Hitting Set Axiom for the class~$\scriptC$} be
the following statement:
\begin{equation*}
\begin{array}{lllll}
  \HS(\scriptC) & := &
  \forall e\ \forall n,d,s,q,r,m{\in}\Log\ (q{\geq}2nd \wedge r{>}m{+}n|q| \rightarrow
  \exists H{=}(\vh_1,\ldots,\vh_r){\in}(S_q^n)^r\ \\
  & & \;\;\;\;\; \forall C{\in}\scriptC_e(n,d,s,m)\
  ((\exists{\av}{\in}\formalZ^n\ C(\av){\not=}0)
\rightarrow (\exists i{<}r\ C(\vh_{i+1}){\not=}0))).
\end{array}
\end{equation*}
It is important to note that the \emph{largeness} hypothesis $r > m+n|q|$
does \emph{not} depend on the representation size $s$; it depends only on
the description size $m$, the bit-complexity $|q|$ of the evaluation points, 
and on the number $n$ of indeterminates. The role of the representation size~$s$
in $\scriptC_e(n,d,s,m)$ is to ensure that the circuits of the class
$\scriptC$ can be evaluated in polynomial time, given $n,d,s$ in unary, a
description $x \in \{0,1\}^m$, and an integer assignment for the
variables. When the definable class~$\scriptC$ is simply the class~$\Circuits$ 
of all circuits, then~$m=s$ and the largeness bound is~$r > s+n|q|$.

We write $\HS(\PV)$ for the axiom-scheme formed by all axioms of the 
form~$\HS(\scriptC)$ where~$\scriptC$ is a definable class of algebraic circuits.
To motivate this notation, we observe that if $g$ is the encoding function
of a definable class $\scriptC$, then the axiom $\HS(\scriptC)$ would be
provably equivalent (in~$\Sonetwo$)
to the following slightly longer but essentially identical variant:
\begin{equation*}
\begin{array}{lllll}
  \HS(g) & := &
  \forall e\ \forall n,d,s,q,r,m{\in}\Log\ (q{\geq}2nd \wedge r{>}m{+}n|q| \rightarrow 
  \exists H{=}(\vh_1,\ldots,\vh_r){\in}(S_q^n)^r\ \\
  & & \;\;\;\;\; \forall x{\in}\{0,1\}^m\ \forall C{\in}\Circuits(n,d,s)\ (g_e(n,d,s,x){=}C \rightarrow \\ 
  & & \;\;\;\;\;\;\;\;\;\;\;\;\; ((\exists{\av}{\in}\formalZ^n\ C(\av){\not=}0)
\rightarrow (\exists i{<}r\ C(\vh_{i+1}){\not=}0)))).
\end{array}
\end{equation*}
Conversely, every axiom~$\HS(g)$ with $g$ in $\PV$ is provably equivalent (in~$\Sonetwo$)
to the axiom~$\HS(\scriptC)$ for the definable class~$\scriptC$ that has as encoding
function the modification of $g$, still in $\PV$, that with parameter~$e$ and input $(n,d,s,x)$
outputs a default trivial circuit that computes the constant~$0$ polynomial
if~$g_e(n,d,s,x)$ is not already
a circuit in~$\Circuits(n,d,s)$. 

Let us make turn this notation into a definition:

\begin{definition} \label{def:hspv}
    For every $\PV$-symbol $g$, the \emph{hitting set axiom
    for $g$} is the formula $\HS(g)$. The \emph{hitting set axiom
    scheme}, denoted by $\HS(\PV)$, is the collection of
    all $\HS(g)$ for all $g \in \PV$.
\end{definition}

An immediate consequence of Theorem~\ref{lem:hittingset} is that~$\HS(\PV)$ is a consequence
of~$\dWPHP(\PV)$ over~$\Sonetwo$. The first goal of this section is to prove
the converse:

\begin{theorem}[Reverse Mathematics of Hitting Sets] \label{thm:equivalence}
  The axiom-schemes~$\dWPHP(\PV)$ and~$\HS(\PV)$ are provably equivalent over~$\Sonetwo$.
\end{theorem}

The next two sections are devoted to the proof of Theorem~\ref{thm:equivalence}.

That $\Sonetwo + \dWPHP(\PV)$ proves $\HS(\PV)$ was shown in Theorem~\ref{lem:hittingset}.
  Here we show the converse. Fix a~$\PV$-symbol~$f$ for which we want to prove~$\dWPHP(f)$. 
  Let~$a$ and~$b$ be such
  that $b \geq 2a\geq 2$, and let $c$ be a setting for the parameters of~$f$.
  Our goal is to show that there exists $y \in [b]$ such
  that for every $x \in [a]$ we have $f_c(x) \not= y$. 
  
  The first step in the proof
  is to apply an amplification transformation to $f_c$ that, without loss of generality,
  will let us assume that $a = 2^m$ and $b = 2^{m+t}$, for $m = |a-1|$ and a~$t = \poly(m)$
  of our choice. This will let us think of $f_c$ as
  a function~$h_e : \{0,1\}^m \to \{0,1\}^{m+t}$. The goal will then become 
  finding a string $y \in \{0,1\}^{m+t}$ that is outside the range of the function $h_e$.
  
  The amplification technique that we need is standard in the 
  context of the weak pigeonhole principle. For the sake
  of completeness, the details of this argument have been made explicit in 
  Section~\ref{app:normalizationandamplification} of the Appendix.
  Concretely, we summarize the composition of Lemmas~\ref{lem:norm} and~\ref{lem:ampl}
  in the following lemma. In its statement, $\PV_2$ refers to a 
  collection of symbols for the (clocked) $\FP^{\NP}$-machines, just as~$\PV$ is 
  a collection of symbols for (clocked)~$\FP$-machines. Indeed, we need symbols for 
  the~$\FP^{\NP}[\wit,q]$-machines:~$\FP^{\NP}$-machines that get witnesses to
  their $\NP$-oracle queries when they are answered~YES,
  and make at most $q$ queries in every computation path. In general, such
  machines compute multi-output functions (i.e., total relations), instead of just
  functions, because their output may depend on 
  the witness-sequence provided by the oracle. The $\FP^{\NP}[\wit,O(\log n)]$-computable 
  multi-output functions
  are $\Sigma^b_2$-definable in $\Sonetwo$ by Theorem~6.3.3 in~\cite{Kra95}.
  
  \begin{lemma}[In~$\Sonetwo$] \label{lem:reduction}
  For every~$\PV$-function~$f$ and every standard natural number~$k \geq 2$,
  there is a~$\PV$-function~$h$ and
  a~$\PV_2[\wit,2]$-function $r$ such that 
  for all $a,b,c$ such that~$b \geq 2a \geq 2$, setting~$m = |a-1|$ 
  and~$e = \langle a,c \rangle$ as parameter for the functions~$h$ and~$r$, 
  the
  restriction of~$h_e$ to~$\{0,1\}^m$ is a function~$h_{e} : \{0,1\}^m \to \{0,1\}^{m^k}$ and, 
  for every~$y \in \{0,1\}^{m^k}$ that is outside the range of~$h_e$,
  there is a computation of~$r_e(y)$ that does not fail, and any such computation
  outputs an element in $[2a] \subseteq [b]$ that is outside the range of~$f_c$ 
  restricted to~$[a]$.
  \end{lemma}

  \begin{proof} 
  Set $t = m^k - m$ and compose the $g$ and $h$ functions, and the $s$ and $r$ functions, 
  in Lemmas~\ref{lem:norm} and~\ref{lem:ampl} of Section~\ref{app:normalizationandamplification}.
  \end{proof}

  The bottom line of this section is that Lemma~\ref{lem:reduction} reduces
  the problem of producing a witness for $\dWPHP^a_b(f_c)$ to the problem of
  producing a string in $\{0,1\}^{m^k}$ that avoids the range of $h_e$. 
  In our application we will choose~$k := 3$, so that
  \begin{equation}
  h_e : \{0,1\}^m \to \{0,1\}^{m^3}. \label{eqn:theh}
  \end{equation}
  The reason for the choice~$k = 3$ will become clear later in the proof. 

\subsection{The Definable Class of Algebraic Circuits} \label{sec:class}

  We continue now with the proof
  of the implication from $\HS(\PV)$ to $\dWPHP^a_b(f_c)$ over~$\Sonetwo$. 
  We argued already that the problem reduces to avoiding the range 
  of~$h_e$ in~\eqref{eqn:theh}.
  
  For the sake of contradiction, assume that for all $y \in \{0,1\}^{m^3}$ there
  is an $x \in \{0,1\}^m$ such that~$h_e(x) = y$. We intend to 
  use the hitting set~$H$ provided by~$\HS(\scriptC)$ for 
  an appropriately chosen definable class
  of algebraic circuits $\scriptC$, with suitable 
  parameters~$e,n,d,s,q,r,m$ 
  defined from $e$ and $m$. The contradiction will come from the following
  three steps of reasoning: 
  \begin{enumerate} \itemsep=0pt
  \item first we define a string~$y \in \{0,1\}^{m^3}$ that 
  suitably encodes the hitting
  set~$H$; 
  \item next we use the assumption to get a shorter $x \in
  \{0,1\}^m$ that can recover~$y$ via $h_e$; 
  \item finally, we use~$x$, \emph{as a compressed version of~$H$,} to build a small 
  non-vanishing algebraic circuit~$C$ in~$\scriptC_e(n,d,s,m)$ which will,
  however, vanish on~$H$, by design. 
  \end{enumerate}
  
  Before we can
  execute this plan first we need to define the suitable class~$\scriptC$
  on which to apply~$\HS(\scriptC)$. We commit to
  the following choice of parameters~$e,n,d,s,q,r,m$ as functions 
  of~$e$ and~$m$. First we note that we are not overloading the 
  names of the parameters:~$e$ and~$m$ will be taken to be themselves.
  Second, since every finite instance of
  $\dWPHP$ of standard size is provable in $\BASIC$, and in particular
  in~$\Sonetwo$, we may assume that $m$ is large enough, i.e., $m \geq
  m_0$ for a fixed standard~$m_0$ to be fixed later. For such
  large $m$, choose
  \begin{equation}
  \begin{array}{llllll}
  n := m & \;\;\; & d := m^2 & \;\;\; & s := m^3 \\
  r := 4m|m| & \;\;\; & q := 2m^3.
  \end{array} \label{eqn:choiceofparams}
  \end{equation}
  Observe that, if $m$ is large enough, then
  \begin{equation}
  \begin{array}{lllllllll}
    m^3 \geq rn|q| & \;\;\; & d \geq 2r|q| & \;\;\; & 
    q \geq 2dn & \;\;\; & r > m+n|q|.
  \end{array}
    \label{eqn:oneequation}
  \end{equation}
  The first part in~\eqref{eqn:oneequation} implies 
  that every triple $(i,j,k) \in [r] \times [n] \times [|q|]$ can be encoded with 
  a number in~$[m^3]$.\footnote{This is the place in proof 
  where the choice $k=3$, so $h_e$ has
  co-domain $\{0,1\}^{m^3}$, is needed.} Fix such an encoding
  of triples and write $e_m(i,j,k) \in [m^3]$ for the encoding of~$(i,j,k)$; 
  we require of this encoding that it is
  provably~$\PV$-computable and provably~$\PV$-invertible, which is easy to achieve in
  many standard ways, given~$m$ as parameter.
  
  We now define the class $\scriptC$. 
  For all $x \in \{0,1\}^m$ and $(i,j,k) \in [r]\times[n]\times[|q|]$, 
  let~$h_e(x; i,j,k)$ denote the $e_m(i,j,k)$-th bit
  of output of~$h_e(x)$. Let~$\mathscr{G}_{e,m} \subseteq \Circuits(n,d)$ 
  be the set of
  algebraic circuits with variables $z_1,\ldots,z_n$ and plugged constants $-1,0,1$
  that have the form
  \begin{equation}
    A_{e,x}(z_1,\ldots,z_n) := \prod_{i \in [r]} \sum_{j \in [n]} \Big({
      z_j - \sum_{k \in [|q|]} h_e(x; i,j,k) \cdot
      2^{k-1}}\Big)^2, \label{eqn:circuit}
  \end{equation}
  where $x \in \{0,1\}^m$. In this description of the algebraic circuit $A_{e,x}$, 
  the $z_1,\ldots,z_n$ are the variables, while the constants $-1,0,1$
  are used for the $h_e(x;i,j,k)$, 
  the~$2$ in~$2^{k-1}$, and 
  the~$-1$ that is implicit in the subtraction. 
  The syntactic individual degree
  of~\eqref{eqn:circuit} on the $z_j$ variables is~$2r$, and the syntactic 
  individual
  degree of~\eqref{eqn:circuit} on the parametric inputs is at most $2r|q|$.
  By the second part in~\eqref{eqn:oneequation} we have~$d \geq 2r|q|$, 
  so indeed, as claimed,~$\mathscr{G}_{e,m}$ is a subset 
  of~$\Circuits(n,d)$. Further, if $m$ is large enough, 
  by the first part in~\eqref{eqn:oneequation} again,
  the representation size of $A_{e,x}$ (as a member of $\Circuits$) is
  bounded by~$m^3$. Thus, by the choice of $s$ in~\eqref{eqn:choiceofparams},
  the set $\mathscr{G}_{e,m}$ is even a subset of $\Circuits(n,d,s)$.
  We let~$\scriptC$ be given by the encoding function $g_e(n,d,s,x)=A_{e,x}$,
  if $n,d,s,m$ obey~\eqref{eqn:choiceofparams} and~$m$ is as determined
  by~$e = \langle a,c \rangle$, i.e.,~$m = |a-1|$;
  otherwise we let~$g_e(n,d,s,x)$ be a default trivial circuit that computes 
  the constant~$0$ polynomial. 
  By construction,~$\scriptC$ is a legitimate definable class of 
  algebraic circuits and~$\scriptC_e(n,d,s,m) := \mathscr{G}_{e,m}$.\footnote{
  Observe that $\scriptC_e(n,d,s,m)$ is a \emph{sparse} class
  since $2^m \ll 2^{m^3} = 2^s$. This is the crucial point
  that makes the proof work. We cannot afford $m=s$ as would be
  the case for $\Circuits(n,d,s)$ 
  because the \emph{representation size} of~$A_{e,x}$ as an explicit circuit
  is roughly $rn|q|$, which is larger than $s$, and 
  no choice of parameters can make~$rn|q|$ smaller than~$s$
  if~$m=s$ and $r$ and~$m$ are to satisfy the requirement that~$r > m+n|q|$
  in the statement of~$\HS$. In contrast, as member of~$\scriptC$,
  the \emph{description size} of~$A_{e,x}$
  is~$m$, which is significantly smaller than~$rn|q| \geq m^2$.}

  \subsection{The Compression Argument} \label{sec:compression}

  We are ready to complete the proof of $\dWPHP(f)$ in $\Sonetwo + \HS(\PV)$.
  Consider~$\HS(g)$ or, equivalently,~$\HS(\scriptC)$, for the class~$\scriptC$ of 
  the previous section with encoding function $g$, and consider 
  the parameters~$e,n,d,s,q,r,m$ chosen in~\eqref{eqn:choiceofparams}. These choices
  are made to guarantee that the hypotheses in 
  the~$\HS(\scriptC)$-axiom are met; concretely, by~\eqref{eqn:oneequation}, 
  the requirement~$q\geq 2nd$ holds and,
  most importantly, the requirement $r > m+n|q|$, which
  ensures the abundance of hitting sets~$H$, 
  also holds.
  
  Let then~$H =
  (\vh_1,\ldots,\vh_r) \in (S_q^n)^r$ be a hitting set
  for~$\scriptC_e(n,d,s,m)$.  Let $y \in \{0,1\}^{m^3}$ be the string
  that encodes $H$ in binary padded with $0$s to length $m^3$. More precisely,
  first note
  that~$H$ can be identified with an element of $[q]^{nr}$, and hence
  with a string in $\{0,1\}^{rn|q|}$, so the first inequality
  in~\eqref{eqn:oneequation} ensures that $m^3$ is big enough to
  encode $H$. Even more explicitly, if we write $\vh_i = (h_{i,1},\ldots,h_{i,n})$
  and write each $h_{i,j} \in S_q = \{0,\ldots,q-1\}$ in
  binary notation as~$h_{i,j} = (h_{i,j,1},\ldots,h_{1,j,|q|}) \in \{0,1\}^{|q|}$, 
  then we choose $y$ such that
  for all triples $(i,j,k) \in [r] \times [n] \times [|q|]$ we have that
  \begin{equation}
    \text{the $e(i,j,k)$-th bit of $y$ equals $h_{i,j,k}$; the rest of bits of $y$
    are $0$.}
    \label{eqn:property}
  \end{equation}
  By the assumption that $\dWPHP$ fails for $h_e : \{0,1\}^m \to \{0,1\}^{m^3}$, 
  there exists an $x \in \{0,1\}^m$
  such that $h_e(x) = y$. In particular, by~\eqref{eqn:property} and
  the definition of $h_e(x;i,j,k)$,
  for all $(i,j,k) \in [r] \times [n] \times [|q|]$
  we have
  \begin{equation}
    h_e(x;i,j,k) = h_{i,j,k}.
    \label{eqn:propertyagain}
  \end{equation}
  Now consider the algebraic circuit~$A_{e,x}(z_1,\ldots,z_n)$ from~\eqref{eqn:circuit}, 
  which is in~$\scriptC_e(n,d,s,m)$ by definition. The polynomial
  computed by $A_{e,x}$ does not evaluate to zero on all $\va \in \formalZ^n$. 
  To see this, take $a_j := 2q$ for $j = 1,\ldots,n$ and note that each factor 
  of the product in $A_{e,x}(a_1,\ldots,a_n)$ is (provably) 
  positive since each
  of its $n$ summands is (provably) positive as all inner sums are less than~$2^{|q|} \leq 2q$. 
  Therefore, since $H$ hits $A_{e,x}$,
  there exists $i \in [r]$ such that $A_{e,x}(\vh_{i}) \not= 0$.  But then
  each factor evaluated at $\vh_i$ is non-zero, and in particular 
  the $i$-th factor is non-zero, which means that there exists
  $j \in [n]$ such that
  \begin{equation}
    h_{i,j} \not= \sum_{k \in [|q|]} h_e(x; i,j,k) \cdot 2^{k-1}.
  \end{equation}
  In turn, since $h_{i,j}$ belongs to~$S_q = \{0,1,\ldots,q-1\}$ and
  we encoded it in binary with~$|q|$ bits, it follows from 
  the (provable) uniqueness of the binary encoding that 
  there exists~$k \in [q]$
  such that~$h_{i,j,k} \not= h_e(x;i,j,k)$. This contradicts~\eqref{eqn:propertyagain},
  and the claim is proved.

  This completes the proof of Theorem~\ref{thm:equivalence}, and this section.

  \subsection{Completeness for Range Avoidance Problems} \label{sec:completeness}

  Consider the total search problem of finding hitting sets for low-degree algebraic circuits:
  
  \begin{quote}
  $\HITTINGSET$: \emph{Given $n,d,s,q,r$ in unary such that $q \geq 2nd$ and $r > s+n|q|$, 
  find a hitting set for $\Circuits(n,d,s)$ over $S_q$ of size at most $r$.}
  \end{quote}
  This is a total function in the complexity class $\TFSigmatwoP$ of total functions in $\FSigmatwoP$; 
  the totality is given by Theorem~\ref{lem:hittingset}.
  We define also its generalization to arbitrary circuit classes. In its
  definition below, the phrase \emph{a description of a class of at most~$2^m$ algebraic
  circuits $\mathscr{G} \subseteq \Circuits(n,d,s)$}
  refers to a Boolean circuit $C$ with $m$ Boolean inputs and $s$ Boolean
  outputs syntactically guaranteed to satisfy $C(x) \in \Circuits(n,d,s)$ 
  for all $x \in \{0,1\}^m$. One way to fulfill the syntactic guarantee is
  to check its output for membership in $\Circuits(n,d,s)$ and to return 
  the code of a default trivial circuit that computes the constant~$0$ polynomial 
  if the check fails.

  \begin{quote}
  $\HITTINGSETCLASSES$: \emph{Given $n,d,s,q,r,m$ in unary such
  that~$q \geq 2nd$ and~$r > m+n|q|$, and given
  a description of a 
  class of at most~$2^m$ algebraic circuits~$\mathscr{G} \subseteq \Circuits(n,d,s)$, 
  find a hitting set for~$\mathscr{G}$ over $S_q$ of size at most~$r$.}
  \end{quote}

  Again Theorem~\ref{lem:hittingset} shows that this is a total function in $\TFSigmatwoP$.
  Indeed, Theorem\ref{lem:hittingset}
  shows that both problems reduce, in polynomial time,
  to the Range Avoidance Problem for Boolean circuits from~\cite{RSW22}. 
  In the terminology of \cite{KKMP21,Korten21}, this is the problem~$\EMPTY$, and
  then both problems belong to the complexity class~$\APEPP$ for which $\EMPTY$ is
  its defining problem. The first of the two
  problems is even in the subclass~$\SAPEPP$ 
  of sparse problems in~$\APEPP$. 

  We observe now that the proof of Theorem~\ref{thm:equivalence} shows
  that the problem about hitting sets for circuit classes 
  is complete for the class $\APEPP$. In
  the precise statement below, recall that a polynomial-time mapping reduction
  between search problems is a pair of polynomial-time computable maps~$f$ and~$r$: 
  the \emph{forward}~$f$ maps any instance $x$ of the first problem to
  an instance~$f(x)$ of the second problem; the \emph{backward}~$r$ maps any
  solution~$y$ to the instance~$g(x)$ of the second problem to a solution $h(x,y)$
  of the first problem. We call this a $\P/\P$-reduction because both
  maps are polynomial-time computable. If the second map is computable
  only by a~$\PtoNP$-machine, then we call it a~$\P/\PtoNP$-reduction.

  \begin{theorem}
  The problem $\HITTINGSETCLASSES$ is complete for the~$\APEPP$ under $\PtoNP$-reductions;
  indeed, $\P/\PtoNP$-reductions suffice.
  \end{theorem}

  \begin{proof}
  Membership in the class follows from the proof of the
  first half of Theorem~\ref{thm:equivalence}; i.e.,
  from the proof of Theorem~\ref{lem:hittingset}. 
  Precisely, that proof shows that our problem reduces to~$\EMPTY$ with a $\P/\P$-reduction.
  Completeness for $\APEPP$ follows from the proof of the second half of 
  Theorem~\ref{thm:equivalence}.
  In more details, we argue that the proof of Theorem~\ref{thm:equivalence} shows how 
  to reduce~$\EMPTY$ to our problem, with a $\P/\PtoNP$-reduction.

  Let an instance~$C$ of $\EMPTY$ be given: $C$ is a Boolean circuit with~$m$ inputs and~$m+1$ outputs.
  Let $a = 2^m$ and $b = 2^{m+1}$, let $c$ be the integer encoding of the Boolean circuit~$C$, 
  and let~$e = \langle a,c \rangle$.
  Let $f$ be the $\PV$-function which, with parameter $c$, computes 
  as the Boolean circuit $C$ on input $x \in \{0,1\}^m$. 
  Let $h$ be the $\PV$-function
  given by Lemma~\ref{lem:reduction} with $k = 3$. Thus $h_e$ is a function as in~\eqref{eqn:theh}.

  Consider the corresponding definable class of algebraic circuits $\scriptC$ 
  from Section~\ref{sec:class} and let $g$ be its encoding function. 
  For the given $m$, consider the setting of parameters $n,d,s,q,r$
  as specified in~\eqref{eqn:choiceofparams}. By~\eqref{eqn:oneequation},
  these parameters satisfy the requirements in the input of our problem.
  By definition we have $g_e(n,d,s,x) = A_{e,x}$ for all~$x \in \{0,1\}^m$; 
  i.e.,~$g_e(n,d,s,x)$ outputs (the code of) the algebraic circuit $A_{e,x}$.
  Let $D$ be the canonical Boolean circuit that computes as~$g_e(n,d,s,x)$ on all 
  inputs~$x \in \{0,1\}^m$. Then $D$ is a description of the class of 
  at most $2^m$ algebraic circuits $\mathscr{G}_{e,m} \subseteq \Circuits(n,d,s)$.
  The proof in Section~\ref{sec:compression} shows that any hitting set $H$ of
  size at most $r$ for
  $\mathscr{G}_{e,m}$ gives a string $y \in \{0,1\}^{m^3}$ that is outside the range 
  of $h_e$ restricted to $\{0,1\}^m$. Concretely, let $y$ be the string specified
  in~\eqref{eqn:property}.
  Finally, let $r$ be the $\PV_2[\wit,2]$-function $r$ from Lemma~\ref{lem:reduction} and
  consider (any) $r_e(y)$. By~\ref{lem:reduction} and the natural correspondence 
  between~$[a]$ and~$\{0,1\}^m$, and between~$[b]$ and~$\{0,1\}^{m+1}$, this is a string 
  in~$\{0,1\}^{m+1}$ that is not in the range of $C$, so a solution 
  to the~$\EMPTY$-instance~$C$. Since the circuit $D$ is computable from $e$ and hence from~$C$
  in time polynomial in the size of $C$, the forward map is computable
  by a polynomial-time machine.
  Since $r$ is computable by a~$\PtoNP$-machine and $m$ is bounded by the size of $C$, 
  the \emph{backward} map 
  is computable by a~$\PtoNP$-machine. Both together give 
  a~$\P/\PtoNP$-reduction and the proof is complete.
  \end{proof}

\appendix

\section{Amplification for dWPHP} \label{app:normalizationandamplification}

  The material of this appendix, or at least the techniques in it, 
  can be considered known. Since we were
  not able to find a reference with the exact result that we need,
  we provide the details here. In bounded arithmetic, 
  the main ideas go back to the seminal
  paper by Paris-Wilkie-Woods~\cite{PWW88}, and have
  been revisited several times; cf. \cite{MacielPitassiWoods2002, 
  Atserias2003} for bounded arithmetic, 
  and~\cite{KKMP21, Korten21}
  for complexity theory. Much earlier, in the area of cryptography, 
  exactly the same construction was used~to build pseudo-random 
  number generators (PRNGs) from hardcore bits~\cite{BlumMicali1982},
  which actually give a kind of guarantee that 
  is incomparable with what we need.
  
  In few words, what is done here is to show
  that the search problem associated to an arbitrary instance of
  the Dual Weak Pigeonhole Principle $\dWPHP^a_b$ with $b \geq 2a$ reduces
  to certain special instances of the same problem.
  In what we call \emph{the normalization step} below, we show how to
  reduce an arbitrary instance with $b \geq 2a$ to an instance 
  of the form~$a = 2^m$ and~$b = 2^{m+1}$.
  In the terminology of Section~\ref{sec:completeness}, this will
  be a $\P/\P$-reduction. In \emph{the amplification step}, also below,
  we show how to reduce an instance of the form~$a = 2^m$ and~$b = 2^{m+1}$,
  to an instance of the form~$a = 2^m$ and~$b = 2^{m+t}$ for any
  arbitrary~$t = \poly(m)$ of our choice. In both cases, we
  do this provably in $\Sonetwo$, as this is what is needed for
  the proof of Theorem~\ref{thm:equivalence}.
  
  To make this section more self-contained, let us recall the definition of
  the Dual Weak Pigeonhole Principle axioms. Let $f$ be a $\PV$-symbol. The axiom $\dWPHP(f)$ 
  is the universal closure of the following formula with free variables $a,b,c$:
  \begin{equation}
  \dWPHP^a_b(f_c) := (b{\geq}2a{\geq}2 \rightarrow \exists y{\in}[b]\ \forall x{\in}[a]\ f_c(x)\not=y).
  \end{equation}
  Here, $f_c(x)$ is alternative notation for $f(x,c)$, thinking of $c$ 
  as a \emph{parameter} for $f$. These parameters may, in particular, specify
  the sizes~$a$ and~$b$ of the intended domain and range of a function~$f_c : [a] \to [b]$,
  but this is not enforced.

  Let $m = |a-1|$ and $n = |b-1|$, so that $n > m$ by $b \geq 2a \geq 2^{m+1}$.
  In particular, all elements in $\{0,\ldots,a-1\}$
  and $\{0,\ldots,b-1\}$ can be represented in 
  binary with~$m$ and $n$ bits, respectively. Subtracting and adding one as appropriate, 
  the same applies to $[a] = \{1,\ldots,a\}$ and $[b] = \{1,\ldots,b\}$.
  
  \subsection{Normalization step}

  In the normalization step we want to
  find a function~$g_d : \{0,1\}^m \to \{0,1\}^{m+1}$, with its parameters~$d$ 
  defined in terms of the
  given~$\PV$-function~$f$ and the parameters $a,b,c$,
  in such a way that, given a point in $\{0,1\}^{m+1}$ that is outside the range of the 
  function~$g_d$, it is easy to find
  a point in $[b]$ that is outside the range of the function~$f_c$ restricted to $[a]$.
  In this case, \emph{easy} means in time polynomial in~$m$ and the length $|c|+|d|$ of
  the parameters. 
  
  Let us note that 
  if~$a$ and~$b$ were 
  both exact powers of two, then it would suffice
  to take~$m=|a-1|$ and let~$g_d$ compute the same as~$f_c$ with all inputs encoded
  in binary notation. 
  In case~$a$ or~$b$ are not exact powers of 
  two, we can still take $m = |a-1|$, but the details of the construction
  require a bit of care to deal with the appropriate encodings of
  the sets. The proof is still straightforward, just more tedious.
  
  \begin{lemma}[In~$\Sonetwo$] \label{lem:norm}
  For every $\PV$-function $f$ there exist $\PV$-functions $g$ and $r$ such that
  the following holds:
  For all $a,b,c$ such that $b \geq 2a \geq 2$, setting $m = |a-1|$ and $d = \langle a,b,c \rangle$,
  the functions~$g_d$ and~$r_d$ have domains and co-domains~$g_d : \{0,1\}^m \to \{0,1\}^{m+1}$ 
  and~$r_d : \{0,1\}^{m+1} \to [2a] \subseteq [b]$, 
  and for all $y \in \{0,1\}^{m+1}$, 
  if $y$ is outside the range of $g_d$ restricted to~$\{0,1\}^{m}$, 
  then $r_d(y)$ is outside the range of $f_c$ restricted to~$[a]$.
  \end{lemma}

  \begin{proof}
  Recall that~$m = |a-1|$ and let $\num_a : \{0,1\}^m \to [a]$ be
  a canonical surjection onto~$[a]$ defined by, say,
  $\num_a(x_1,\ldots,x_m) = \sum_{i=1}^m x_i 2^{i-1}\;\mathrm{mod}\;a$,
  with the residue classes~$\mathrm{mod}\;a$ naturally identified with 
  the elements of the set~$[a]$.
  Fix a corresponding easily computable inverse function~$\num_a^{-1} : [a] \to \{0,1\}^m$ 
  satisfying the invertibility condition
  \begin{equation}
  \num_a(\num_a^{-1}(y)) = y \label{eqn:invertibilityofnum}
  \end{equation}
  for all $y \in [a]$.
  Similarly, but dually, note that $|2a-1|=m+1$, so
  let $\bin_{2a} : [2a] \to \{0,1\}^{m+1}$ be a canonical surjection
  onto $\{0,1\}^{m+1}$, with corresponding inverse $\bin_{2a}^{-1} : \{0,1\}^{m+1} \to [2a]$
  satisfying the invertibility condition
  \begin{equation}
  \bin_{2a}(\bin^{-1}_{2a}(\vy)) = \vy \label{eqn:invertibilityofbin}
  \end{equation}
  for all $\vy \in \{0,1\}^{m+1}$.
  All these are $\PV$-functions and the invertibility conditions~\eqref{eqn:invertibilityofnum}
  and~\eqref{eqn:invertibilityofbin} are provable in $\Sonetwo$.

  Let $g_d : \{0,1\}^m \to \{0,1\}^{m+1}$ be the function
  defined on $\vx \in \{0,1\}^m$ by
  \begin{equation}
  g_d(\vx) := \bin_{2a}(f_c(\num_{a}(\vx))\;\mathrm{mod}\;2a).
  \label{eqn:g}
  \end{equation}
  We claim that, to find a $y \in [b]$ such that~$f_c(x) \not= y$ for all~$x \in [a]$, it suffices
  to find a~$\vy \in \{0,1\}^{m+1}$ such that~$g_d(\vx) \not= \vy$ for all~$\vx \in \{0,1\}^m$.
  Indeed, given such a~$\vy \in \{0,1\}^{m+1}$, just 
  let~$y = r_d(\vy) := \bin^{-1}_{2a}(\vy) \in [2a] \subseteq [b]$. 
  To prove the correctness, note that
  any~$x \in [a]$ such that~$f_c(x) = y$ would give~$\vx = \num^{-1}_a(x) \in \{0,1\}^m$
  such that
  \begin{align*}
  g_c(\vx) & = g_c(\num^{-1}_a(x)) = \bin_{2a}(f_c(\num_a(\num^{-1}_a(x))) \;\mathrm{mod}\; 2a) = \\
  & = \bin_{2a}(f_c(x)\;\mathrm{mod}\;2a) = \bin_{2a}(y\;\mathrm{mod}\;2a) = \bin_{2a}(y) = \\
  & = \bin_{2a}(\bin^{-1}_{2a}(\vy)) = \vy,
  \end{align*}
  where the first equality follows from the choice of $\vx$, the second follows
  from~\eqref{eqn:g}, the third follows from~\eqref{eqn:invertibilityofnum}
  on the fact that~$x \in [a]$, the fourth follows from the assumption that~$f_c(x) = y$,
  the fifth follows from the fact that~$y \in [2a]$, the sixth follows from the choice
  of~$y$, and the last follows from~\eqref{eqn:invertibilityofbin} on 
  the fact that~$\vy \in \{0,1\}^{m+1}$.
  \end{proof}

  \subsection{Amplification step} 

  In the amplification step we extend
  the co-domain of the function~$g_d : \{0,1\}^m \to \{0,1\}^{m+1}$ 
  to the set~$\{0,1\}^{m+t}$, for any desired~$t = \poly(m)$. The construction 
  of this new function~$h_d : \{0,1\}^m \to \{0,1\}^{m+t}$ will be such
  that, given a point outside the range of~$h_d$, it will be possible to
  get a point outside the range of~$g_c$. 
  In this case, the sense in which it is \emph{easy} to translate the solution
  is a bit more nuanced since we will 
  need the help of an~$\NP$-oracle in the computation. Concretely,
  the translation function will be represented by a $\PV_2$-symbol
  that computes a multi-output function (i.e., a total relation) 
  in $\FPtoNP[\wit,2]$. Here, $\FPtoNP[\wit,q]$
  denotes the class of multi-output functions that can be computed by $\FPtoNP$-machines
  that get witnesses to their $\NP$-oracle queries when they are answered YES, and
  make at most $q$ queries in every computation path. For $q = O(\log n)$, where
  $n$ is the size of the input, the computations of such machines are $\Sigma^b_2$-definable
  in $\Sonetwo$ by Theorem~6.3.3 in~\cite{Kra95}.

  \begin{lemma}[In~$\Sonetwo$] \label{lem:ampl}
  For every $\PV$-function $g$ there exist
  a $\PV$-function $h$ and a $\PV_2[\wit,2]$-function $s$ such that the following 
  holds: For all $d$ and
  all $m,t \in \Log$,
  setting $e = \langle d,m,t \rangle$ we have that $h_e$ and $s_e$ compute functions
  $h_e : \{0,1\}^m \to \{0,1\}^{m+t}$
  and (multi-output)~$s_e : \{0,1\}^{m+t} \to \{0,1\}^{m+1}$, and for all $y \in \{0,1\}^{m+t}$,
  if $y$ is outside the range of $h_e$ restricted to $\{0,1\}^m$, then there is a computation
  of $s_e(y)$ that does not fail, and any such computation 
  outputs a string outside the range of $g_d$ restricted to $\{0,1\}^m$.
  \end{lemma}

  We begin the proof of Lemma~\ref{lem:ampl} by fixing notation to manipulate strings.
  For a string~$x = (a_1,\ldots,a_k) \in \{0,1\}^k$ and
  indices~$i,j \in [k]$ 
  we write~$x[i,j]$ to denote the substring~$(a_i,a_{i+1},\ldots,a_j)$
  of $x$ between positions~$i$ and~$j$ with endpoints included (unless~$i > j$). This is
  a string in $\{0,1\}^{j-i+1}$ if $i \leq j$, and the empty string if~$i > j$.
  For strings~$x = (a_1,\ldots,a_k) \in \{0,1\}^k$ 
  and~$y = (b_1,\ldots,b_\ell) \in \{0,1\}^\ell$, we write~$x{:}y$ 
  to denote the concatenation~$(a_1,\ldots,a_k,b_1,\ldots,b_\ell)$
  of~$x$ and~$y$, which is a string in $\{0,1\}^{k+\ell}$.

  With this notation, we define the function $h_e$. 
  For $i = 0,1,\ldots,t$, let $h_{e,i} : \{0,1\}^{m} \to \{0,1\}^{m+i}$
  be defined for all~$x \in \{0,1\}^{m}$ 
  by the following recursion:
  \begin{equation*}
  \begin{array}{lllll}
  h_{e,0}(x) & = & x \\
  h_{e,1}(x) & = & g(x) \\
  h_{e,i+2}(x) & = & g_d(h_{e,i+1}(x)[1,m]){:}h_{e,i+1}(x)[m+1,m+i+1].
  \end{array}
  \end{equation*}
  Note that for any $t = \poly(m)$, the function $h_{e,i} : \{0,1\}^m \to \{0,1\}^{m+t}$
  is computable in polynomial time. Let $h$ be the corresponding $\PV$-symbol in such
  a way that its parameterization by $e,i$ is~$h_{e,i}$, and its parameterization by $e$ alone
  is $h_{e,t}$.
  We claim that, to find a~$y_0 \in \{0,1\}^{m+1}$ such that~$g(x_0) \not= y_0$ for
  all~$x_0 \in \{0,1\}^m$, it suffices to find a~$y \in \{0,1\}^{m+t}$ such 
  that~$h_e(x) \not= y$ for all~$x \in \{0,1\}^m$. Consider the following stronger claim:

  \begin{claim}[In~$\Sonetwo$] \label{lem:thatlemma}
  For all $m,i \in \Log$,
  all $y \in \{0,1\}^m$, all $v \in \{0,1\}^i$, and all $z \in \{0,1\}^m$, 
  if~$h_{e,i}(z) = y{:}v$, then~$h_{e,i+1}(z) = g(y){:}v$.
  In particular, for all~$y_0 \in \{0,1\}^{m+1}$ and all~$v \in \{0,1\}^i$,
  if the concatenated 
  string~$y_0{:}v$ is outside the range of~$h_{e,i+1}$, then either~$y_0$ is 
  outside the range of~$g_d$, or any~$w \in \{0,1\}^m$ that witnesses
  the opposite by satisfying~$g_d(w) = y_0$ is such that the concatenated
  string~$w{:}v$ is outside the range of~$h_{e,i}$.
  \end{claim}

  \begin{proof}
  The first statement is a direct consequence of the recursive definition of the~$h_{e,i}$.
  To prove the second statement, assume that there exists~$w \in \{0,1\}^m$
  such that~$g_d(w) = y_0$ yet~$w{:}v$ is inside 
  the range of~$h_{e,i}$, say~$h_{e,i}(u) = w{:}v$
  for some $u \in \{0,1\}^m$.
  Then, by the first part of the lemma we get~$h_{e,i+1}(u) = g_d(w){:}v = y_0{:}v$, which
  means that $y_0{:}v$ is inside the range of~$h_{e,i+1}$.
  \end{proof}
  
  In case $i = 0$, the second part of Lemma~\ref{lem:thatlemma}
  concludes that $y_0$ is outside the range of~$h_{e,1}$, which equals~$g_d$.
  To see this observe that if $i = 0$, then $v$ is the empty string, 
  and any $w \in \{0,1\}^m$ 
  is always \emph{inside} the range of $h_{e,0}$, since $h_{e,0}$ is the identity map on $\{0,1\}^m$. 
  This means that, 
  when given a $y \in \{0,1\}^{m+i}$ that
  is outside the range of $h_{e,i}$ with $i > 0$, the following~$\FPtoNP[\wit]$ procedure (for now
  making more than~$2$ queries)
  halts in less than $i$ iterations and
  finds a $y_0 \in \{0,1\}^{m+1}$ that is outside the range of~$g_d$:

  \medskip
  \begin{tabbing} \itemsep=0pt
  \hspace{0.4cm} \= 1. \hspace{0.1cm} \= given $y = (a_1,\ldots,a_{m+i}) \in \{0,1\}^{m+i}$, \\
  \> 2. \> set $y_0 := (a_1,\ldots,a_{m+1})$, \\
  \> 3. \> for $j = 0,1,\ldots,i-2$ do the following: \\
  \> 4. \> \hspace{0.4cm} \= query ``$\exists w{\in}\{0,1\}^m\ g_d(w){=}y_j$'' to the $\NP$-oracle \\
  \> 5. \> \> if answer is NO, halt and output~$y_j$, \\
  \> 6. \> \> if answer is YES, get such~$w$ and set~$y_{j+1} := w{:}a_{m+j+2}$, \\
  \> 7. \> halt and fail.
  \end{tabbing}
  \medskip
  
  \noindent 
  In turn, this $\FPtoNP[\wit]$-procedure is special in that it halts after it gets the first NO.
  This means that it can be replaced by the following 
  different but equivalent $\FPtoNP[\wit,2]$-procedure:

    \medskip
  \begin{tabbing} \itemsep=0pt
  \hspace{0.4cm} \= 1. \hspace{0.1cm} \= given $y = (a_1,\ldots,a_{m+i}) \in \{0,1\}^{m+i}$, \\
  \> 2. \> set $y_0 := (a_1,\ldots,a_{m+1})$, \\
  \> 3. \> query ``$\exists k{<}i{-}1\ \exists y_1,\ldots,y_k\ \exists w_0,\ldots,w_{k-1}\ \forall j{<}k\ (g_d(w_j)=y_j \wedge y_{j+1}=w_j{:}a_{m+j+2})$'' \\
  \> 6. \> if answer is NO, halt and fail, \\
  \> 7. \> if answer is YES, get such~$k,y_1,\ldots,y_k,w_0,\ldots,w_{k-1}$, \\
  \> 8. \> query ``$\exists w_{k}\ g_d(w_{k}){=}y_{k}$'' \\
  \> 9. \> if answer is NO, halt and output $y_k$, \\
  \> 10. \> halt and fail.
  \end{tabbing}
  \medskip

  We now show that the desired claim follows.
  Let $s_{e,i}$ denote the $\FP^{\NP}[\wit,2]$-machine 
  defined above; precisely, $s$ is a~$\PV_2[\wit,2]$-symbol
  in the theory, and $s_{e,i}$ is its parameterization with $e,i$. 

  \begin{claim}[In~$\Sonetwo$]
  For all~$m,i \in \Log$ such that $i > 0$ and 
  all~$y \in \{0,1\}^{m+i}$, if~$y \in \{0,1\}^{m+i}$ 
  is outside the range of~$h_{e,i}$, then 
  there exists a computation of $s_{e,i}(y)$ that does not fail,
  and any such computation outputs a string in $\{0,1\}^{m+1}$ that is 
  outside the range of~$g$. 
  \end{claim}

  \begin{proof}
  For fixed positive $m,i \in \Log$ and fixed $y = (a_1,\ldots,a_{m+i}) \in \{0,1\}^{m+i}$
  the strategy for this proof is a maximization argument
  showing that a longest computation of $s_{e,i}(y)$ exists. We
  exploit the maximality to argue that this computation
  does not fail and also that any such computation outputs a correct solution.
  The maximization argument is proved by using the \emph{length-maximization} 
  principle $\Sigma^b_1$-$\lengthmax$, which is available in~$\Sonetwo$ (see Lemma~5.2.7 
  in~\cite{Kra95}).
  
  Consider the statement $\phi(x)$ 
  asserting of $k = |x|$ that 
  there exist~$y_0,y_1,\ldots,y_k \in \{0,1\}^{m+1}$
  and~$w_0,w_1,\ldots,w_{k-1} \in \{0,1\}^m$ that witness
  the first oracle query in the algorithm for $s_{e,i}$.
  Formally:
  \begin{equation}
  \begin{array}{llllll}
  \phi(x) & := & \exists k{\leq}x\ (k{=}|x|\wedge \exists y_0,y_1,\ldots,y_k{\in}\{0,1\}^{m+1}\
  \exists w_0,w_1,\ldots,w_{k-1}{\in}\{0,1\}^m\ \\
  & & \;\;\;\;\;\; y_0{=}y[1,m+1] \wedge 
  \forall j{<}k\ (g_d(w_j){=}y_j \wedge y_{j+1}{=}w_j{:}a_{m+j+2})).
  \end{array}
  \end{equation}
  This is a $\Sigma^b_1$-formula with $m,i,y=(a_1,\ldots,a_{m+i})$ as parameters: 
  the sequences in it have the length of a 
  length, and the $\forall j{<}k$ quantifier in it
  is sharply bounded (by $k = |x|$). We intend to apply the length-maximization
  principle~$\phi$-$\lengthmax$ with bound $2^{i-1}-1$. 
  Recall that $i>0$. The existence of a 
  length-maximum $x \leq 2^{i-1}-1$ such that~$\phi(x)$ holds will let us 
  argue that there is a computation of~$s_{e,i}(y)$ that does not fail,
  and also that any such computation
  outputs a string in~$\{0,1\}^{m+1}$ that is outside the range of $g_d$.
  
  First, note that~$\phi(0)$ holds by
  setting~$k=0$ (recall that $|0|=0$)
  and~$y_0 = y[1,m+1]$, and by setting~$w_0,w_1,\ldots,w_{k-1}$ to
  the empty sequence of strings; the quantifier $\forall i{<}k$
  holds vacuously in this case. By~$\phi$-$\lengthmax$ applied to the upper bound $2^{i-1}-1$,
  there exists a length maximum~$x \leq 2^{i-1}-1$ such that $\phi(x)$ holds.
  Let $k = |x|$ and note that~$k \leq i-1$ since every number below~$2^{i-1}$ 
  has length~$i-1$ or less.
  By the definition of the procedure~$s_{e,i}$ and the maximality of $x \leq 2^{i-1}-1$,
  which could have length up to $i-1$,
  if $k < i-1$ and $y_0,y_1,\ldots,y_k$ are the witnesses for $\phi(x)$, 
  then $s_i(\vy)$ does not fail and outputs $y_k$. In addition, in such a case 
  the $\NP$-oracle answered~NO on the query ``$\exists w_k{\in}\{0,1\}^m\ g(w_k){=}y_k$'', 
  and hence~$s_{e,i}(y) = y_k$ is the string outside the range of $g_d$ we were
  looking for. To complete the proof it remains to be seen 
  that, indeed,~$k < i-1$.

  Suppose the contrary, so $k = i-1$ and hence $|x|=i-1$. Let $y_0,y_1,\ldots,y_{i-1}$
  and $w_0,w_1,\ldots,w_{i-2}$ be the witnesses for $\phi(x)$.
  For $j = 0,1,2,\ldots,i-1$, let $v_j = (a_{m+j+2},\ldots,a_{m+i})$; note
  that $v_j$ has length $i-j-1$ and that $v_{i-1}$ is the empty string. 
  Consider the sequence
  of concatenations $y_j{:}v_j$ for $j = 0,1,2,\ldots,i-1$. Note that $y_0{:}v_0$
  is $y$. By assumption $y$ is outside the range of $h_{e,i}$.
  Since $y_0$ is in the range of $g_d$ and $w_0$ witnesses it, 
  by the second part of Lemma~\ref{lem:thatlemma} we conclude that
  $w_0{:}v_0$ is outside the range of $h_{e,i-1}$. Now note that $w_0{:}v_0$
  equals $y_1{:}v_1$, so we have shown that $y_1{:}v_1$ is outside the range 
  of $h_{e,i-1}$. 
  More generally, consider the $\Pi^b_1$-formula
  \begin{equation}
  \psi(x) := \forall j{\leq}x (j{=}|x|\wedge j{\leq}i{-}1 \rightarrow 
  \forall w{\in}\{0,1\}^m\ h_{e,i-j}(w)\not= y_j{:}v_j),
  \end{equation}
  with $i$ among others as parameter.
  This says of the length $j:=|x|$ that $y_j{:}v_j$ is outside the range of $h_{e,i-j}$.
  We know that $\psi(0)$ holds by assumption since~$y = y_0{:}v_0$
  and~$y$ is outside the range of $h_{e,i}$ (and recall $|0|=0$).
  Further, the second statement in 
  Lemma~\ref{lem:thatlemma} combined with the fact that $g_d(w_j)=y_j$
  holds for all $j<i-1$ shows that~$\psi(\floor{x/2})$ implies $\psi(x)$, for all $x$. 
  By $\Pi^b_1$-$\pind$ we get that~$\psi(2^{i-1}-1)$ holds.
  But $v_{i-1}$ is the empty string, so~$\psi(2^{i-1}-1)$
  says that $y_{i-1}$ is outside the range of $h_{e,1}$, which is absurd since
  $y_{i-1} = g_d(w_{i-2})$ and $h_{e,1}$ is~$g_d$.
  \end{proof}

  \paragraph{Acknowledgments.} First author is also
  affiliated with Centre de Recerca Matem\`atica, Bellaterra, Barcelona,
  and is partially supported by grant PID2022-138506NB-C22 (PROOFS BEYOND)
  and Severo Ochoa and Mar\'{\i}a de Maeztu Program for Centers and Units of Excellence 
  in R\&D (CEX2020-001084-M), funded by AEI. The second author is partially supported by 
  the European Research Council (ERC) under the European Union's Horizon 2020 research and 
  innovation programme (grant agreement No.~101002742). This work was initiated while both 
  authors were on residence at the Simons Institute for the
  Theory of Computing, Berkeley, CA, in Spring 2023. 

We would like to thank Emil Je{\v{r}}{\'a}bek and Jiatu Li for  helpful discussions regarding this project. The TikZ code in this paper is modified from the original 
due to Jan Hlavacek, stackexchange member,
under license CC BY-SA 2.5.


\end{document}